\documentclass[onecolumn, conference]{IEEEtran}

\usepackage{color}
\usepackage{times}
\usepackage{amssymb}
\usepackage{amsmath}
\usepackage{epsfig}
\usepackage{cite}
\usepackage{verbatim}
\usepackage{enumerate}
\usepackage{subfigure}
\usepackage{multicol}
\usepackage{pstricks}
\usepackage{pst-node}
\usepackage{amsmath}
\usepackage{graphicx}
\usepackage{setspace}
\usepackage{algpseudocode}
\usepackage{algorithm}
\usepackage{cases}

\newtheorem{property}{\bf Property}

\newtheorem{lemma}{\bf Lemma}

\newtheorem{definition}{\bf Definition}
\newtheorem{remark}{\bf Remark}

\newcommand{\define}{\stackrel{\triangle}{=}}
\newcommand{\fbeta}{\vec{\phi}}
\newcommand{\gbeta}{\phi}
\newcommand{\lbeta}{\vec{\chi}}
\newcommand{\LL}{\mathcal{L}}
\newcommand{\MM}{\mathcal{M}}

\begin{document}
\title{Optimal Repair of {MDS} Codes in Distributed Storage via Subspace Interference Alignment}
\author{\authorblockN{Viveck R. Cadambe, Syed A. Jafar}
\authorblockA{Electrical Engineering and Computer Science\\
University of California Irvine, \\
Irvine, California, 92697, USA\\
Email: \{vcadambe, syed\}@uci.edu} \and
\authorblockN{ Cheng Huang, Jin Li}
\authorblockA{Communications, Collaborations and Systems Group,\\
Microsoft Research, Redmond, WA, USA\\
Email: \{chengh, jinl\}@microsoft.com}}

{\let\thefootnote\relax\footnotetext{ {This paper will be published, in part, in the Proceedings of IEEE Symposium on Information Theory (ISIT) 2011 \cite{Cadambe_Huang_Li_Permutation_ISIT}.}}}

\maketitle
\begin{abstract}
It is well known that an $(n,k)$ code can be used to store information in a distributed storage system with $n$ nodes/disks. If the storage capacity of each node/disk is normalized to one unit, the code can be used to store $k$ units of information, where $n>k$.  If the code used is maximum distance separable (MDS), then the storage system can tolerate up to $(n-k)$ disk failures (erasures), since the original information can be reconstructed from \emph{any} $k$ surviving disks. The focus of this paper is the design of a systematic MDS code with the additional property that a single disk failure can be repaired with minimum \emph{repair bandwidth}, i.e., with the minimum possible amount of data to be downloaded for recovery of the failed disk. Previously, a lower bound of $\frac{n-1}{n-k}$ units has been established by Dimakis et. al, on the repair bandwidth for a single disk failure in an $(n,k)$ MDS code based storage system, where each of the $n$ disks store $1$ unit of data. Recently, the existence of asymptotic codes achieving this lower bound for arbitrary $(n,k)$ has been established by drawing connections to an asymptotic interference alignment scheme developed by Cadambe and Jafar for the interference channel. While the recent asymptotic constructions show the existence of codes achieving this lower bound in the limit of large code sizes, finite code constructions achieving this lower bound existed in previous literature only for the special (high-redundancy) scenario where $k \leq \max(n/2,3)$. The question of existence of finite codes for arbitrary values of $(n,k)$ achieving the lower bound on the repair bandwidth remained open. As a main contribution of this paper, we provide the first known construction of a finite code for {arbitrary} $(n,k)$, which can repair a single failed \emph{systematic} disk by downloading exactly $\frac{n-1}{n-k}$ units of data. The codes, which are optimally efficient in terms repair bandwidth are based on permutation matrices\footnote{The permutation marix based constructions of this paper have been discovered in parallel by Tamo et. al in \cite{Bruck_etal_ISIT}}. We also show that our code has a simple repair property which enables efficiency, not only in terms of the amount of repair bandwidth, but also in terms of the amount of data accessed on the disk. We also generalize our permutation matrix based constructions by developing a novel framework for repair-bandwidth-optimal MDS codes based on the idea of subspace interference alignment - a concept previously introduced by Suh and Tse the context of wireless cellular networks.

\end{abstract}

\newpage
\section{Introduction}
Consider a distributed storage system with $n$ distributed data disks, with each disk storing one unit of data. Assume that the amount of information to be stored in this storage system is equal to $k$ units, where $k < n,$ with the extra storage space of $n-k$ units used to build redundancy in the system. Then, it is well known that the optimal tolerance to failures, for a fixed amount of storage, can be provided by using a $(n,k)$ maximum distance separable (MDS) erasure code to store the data. Such a code would tolerate \emph{any} $(n-k)$ disk failures (erasures), since the MDS property ensures that the original information can be recovered by using \emph{any} $k$ surviving disks. When disk failure occurs, efficient (fast) recovery of the failed disk/s is important, since replacing the failed disk/s before other disks fail reduces the chance of data loss and improves the overall reliability of the system. While an MDS code based storage system can tolerate a worst-case failure scenario of $k$ disks, the most common failure scenario in a storage system is the case where a single disk fails. A problem that has received considerable attention in recent literature \cite{Dimakis_Godfrey_Wainwright_Ramachandran, Wu_Dimakis, Shah_etal, Suh_Ramachandran, Cadambe_Jafar_Maleki, Suh_Ramachandran_general, Gaston_Pujol_Circulant, Dimakis_Survey, Cullina_Dimakis_Ho, Dimakis_EVENODD, RDP_repair}, and is the focus of this paper, is the recovery efficiency (speed) of a single disk failure in an MDS code for distributed storage systems. The increased interest in efficient repair for erasure codes stems, in part, from connections of the problem to various important topics in information and coding theory. First, naturally, the problem is related to the classical field of erasure coding for storage. Second, as demonstrated in \cite{Dimakis_Godfrey_Wainwright_Ramachandran}, the problem of efficient repair connected to network coding. In particular, it is connected to the multi-source network coding problem with generalized demands (i.e., non-multicast) - a classical open problem in network information theory. Finally, as demonstrated in references \cite{Wu_Dimakis, Shah_etal, Suh_Ramachandran}, it is connected to the interference management strategy of interference alignment  - technique widely studied in the context of wireless communications. This final connection will be especially explored in detail in this paper.



When a single node fails in a storage system, a new node enters the storage system, connects to the surviving $d=n-1$ disks via a network, downloads data from these $(n-1)$ surviving disks, and reconstructs the (data stored in the) failed node\footnote{In this paper, we restrict ourselves to \emph{exact repair}, where the new node has to be a replica of the failed node. Note that this is unlike \cite{Dimakis_Godfrey_Wainwright_Ramachandran, Wu_Explicit} which consider \emph{functional repair} where the new node only has to be information equivalent to the failed node}. The primary factor in determining the speed of recovery of the failed node is the amount of time taken for the new node to download the necessary data from surviving disks, which, in turn, depends on the amount of data accessed and downloaded\footnote{There is a subtle difference between the amount of data \emph{accessed} and \emph{downloaded}; Such differences are explored later in Section \ref{sec:access}} by the new node. This problem has been studied from the perspective of the amount of data to be {downloaded} - also known as the \emph{repair bandwidth} - by the new node for successful recovery of the failed node in \cite{Dimakis_Godfrey_Wainwright_Ramachandran, Wu_Dimakis, Shah_etal, Suh_Ramachandran, Cadambe_Jafar_Maleki, Suh_Ramachandran_general, Gaston_Pujol_Circulant, Dimakis_Survey, Cullina_Dimakis_Ho}. Note that a trivial repair strategy for any $(n,k)$ MDS code is to achieve a repair bandwidth of $k$ units for a single failed disk. This is because the entire original data, and hence the failed disk, can be recovered with the new node reading any set of $k$ surviving disks completely. A natural question of interest is the following: what is the minimum repair bandwidth required for a single failed node in an MDS code based distributed storage system? A cut-set lower bound for this question, i.e., for the minimum repair bandwidth, was derived to be $\frac{n-1}{n-k} <k $ units in reference \cite{Dimakis_Godfrey_Wainwright_Ramachandran}. The question of whether this lower bound is achievable via code constructions has received considerable attention in recent literature \cite{Wu_Dimakis, Shah_etal, Suh_Ramachandran, Cadambe_Jafar_Maleki, Suh_Ramachandran_general, Gaston_Pujol_Circulant, Dimakis_Survey}. In particular, recent literature has made progress on this problem of minimum repair bandwidth for repair by drawing connections to the wireless interference management technique of interference alignment. Related results in current literature related to this problem is summarized below.
\begin{enumerate}
\item \emph{Finite Codes for Low Rates:} By connecting the problem of exact repair to the wireless interference management technique of interference alignment, codes which achieve the repair-bandwidth lower bound of $\frac{(n-1)}{n-k}$ have been found in \cite{Wu_Dimakis, Suh_Ramachandran, Shah_etal, Cullina_Dimakis_Ho} for the case where $k \leq \max(n/2,3)$. In other words, if the \emph{rate}, $k/n$, of the code is smaller or equal to than half, finite explicit MDS code constructions exist which can repair a failed node with a repair bandwidth of $\frac{n-1}{n-k}$ units. The repair bandwidth was achieved with the new node downloading $\frac{1}{n-k}$ units of each of the $n-1$ surviving nodes. 
\item \emph{Asymptotic Codes for Arbitrary $(n,k)$:} For arbitrary $(n,k)$ references \cite{Cadambe_Jafar_Maleki, Suh_Ramachandran_general} used the asymptotic interference alignment scheme constructed in reference \cite{Cadambe_Jafar_int}, in the context of wireless interference channels, to generate codes which achieve the optimal repair bandwidth of $\frac{(n-1)}{n-k}$ asymptotically as the size of the code becomes arbitrarily large. 
\end{enumerate}
While the above results are interesting from a theoretical perspective, a matter of relevance for several storage systems in practice are efficient repair strategies for high rate codes, i.e., for storage systems that have a small number of parity nodes as compared to the number of systematic codes and therefore operate in the regime where $k/n > 1/2$. While this asymptotic constructions provide an interesting theoretical limit to what practical codes can achieve, the existence of \emph{finite} codes achieving a repair bandwidth of $\frac{n-1}{n-k}$ units remained an open problem of practical interest. In fact, for arbitrary $(n,k)$, the construction of finite codes having a repair strategy more efficient than the trivial repair strategy with a repair bandwidth of $k$ units remained open. It is this open problem that is the main focus of this paper. We shall next take a closer look at this open problem from the perspective of literature and techniques associated with interference alignment.

\subsection{Connections of Repair Bandwidth to Interference Alignment}
In the context of linear codes (which suffices for this paper), the connections between exact repair and interference in wireless systems can be understood as follows. Consider an $(n,k)$ systematic code, where the first $k$ nodes are systematic and hence store $k$ (uncoded) independent sources, each of size one unit. The remaining $n-k$ nodes are parity nodes. Each parity node stores a linear combination of the $k$ sources, where the combinations are defined by the code generator matrix. Now, suppose that a node, say the first node, fails. In order to repair this node, we assume that the new node downloads a certain set of linear combinations from each of the $n-1$ surviving nodes. The goal is to recover the first source from this set of linear combinations. The $k-1$ surviving systematic nodes store information that is independent of the first source. The information of this first source is stored in the $n-k$ parity nodes - but this desired information in the parity nodes is ``mixed'' with the remaining $(k-1)$ sources corresponding to the remaining $k-1$ systematic nodes. These $k-1$ sources which are not required by the new node, but arrive in the linear combinations downloaded from the parity nodes because they are ``mixed'' with the first source are analogous to interference in wireless communication systems. The coding matrices, which define how the sources are mixed into parity nodes, are analogous to channel matrices in wireless communications which also perform the same function. The linear combinations downloaded by the new node to repair the failed node are analogous to the beamforming vectors in wireless communications (See \cite{Suh_Ramachandran, Suh_Ramachandran_general} for instance). In both applications, the greater the extent of alignment, the more efficient is the system. In the wireless context, interference alignment reduces the footprint of the interference at a receiver and frees a greater number of dimensions for the desired signal (and typically leads to improved number of \emph{degrees of freedom} \cite{Cadambe_Jafar_int}). In the repair context, interference alignment reduces the footprint of the interfering sources at the new node, and hence means that a smaller number of units need to be downloaded to cancel this interference. However, one important difference exists - in the wireless context, the channel matrices are given by nature and cannot be controlled, whereas, in the storage context, the coding matrices are a design choice. 

The approach of references \cite{Cadambe_Jafar_Maleki, Suh_Ramachandran} in asymptotic code construction essentially stemmed from mimicking the wireless interference channel matrices in code construction. These references used diagonal coding sub-matrices analogous to those obtained using symbol extensions and vector coding in wireless channels without inter-symbol-interference. The surprising insight of these references is that, even though there is additional freedom in the storage context as compared to the wireless context because the coding matrices can be designed, the cut-set lower bound can be achieved asymptotically by mimicking the wireless channel matrices for coding in the storage context. In other words, there is no loss from the perspective of the extent of alignment, in an asymptotic sense, when the wireless channel matrices are used for coding in the storage context. Because the coding matrices are analogous to the channel matrices in wireless context, the size of the code is similar to the size of the channel matrices (or the symbol extensions used). In the wireless context of naturally occurring channel matrices, asymptotically large channel matrices (and more generally, asymptotically large amount of diversity) is necessary in general to achieve the maximum extent of alignment, at least, with linear schemes \cite{Cadambe_Jafar_int, Bresler_Tse_diversity}. However, the existence of finite codes for storage is related to the following question: if we have the freedom to design these coding (channel) matrices, can we achieve the desired extent of alignment with finite-size matrices, or are asymptotic schemes unavoidable much like the wireless context? It is worth noting that literature in interference alignment contains examples of wireless channels with certain special channel matrices, where, interference alignment is indeed achieved with finite-size channel matrices \cite{Cadambe_Jafar_int, Nazer_Gastpar_Jafar_Vishwanath, Suh_Tse_subspace}. Of relevance to this work is reference \cite{Suh_Tse_subspace} which shows that if the channel matrices have a specific tensor (Kronecker) product structure, then alignment is possible with finite-size channel matrices using the notion of subspace interference alignment. While these examples serve the purposes of simplifying the concept of alignment for exposition, their practical applicability in the wireless context is limited, because of the nature of the wireless channel. In the storage context, however, the coding (channel) matrices are a design choice; in this paper, we exploit this flexibility and the insights of interference alignment literature (and reference \cite{Suh_Tse_subspace}, in particular) to develop finite-size code constructions for distributed storage.

Before we proceed, we note that there exists, in literature, a parallel line of work, which studies the repair bandwidth for codes which are not necessarily MDS and hence use a greater amount of storage for a given amount of redundancy \cite{Dimakis_Godfrey_Wainwright_Ramachandran, Shah_etal, Shah_insufficiency, Rouayheb_Ramchandran_Fractional}. These references study the trade-off between the amount of storage and the repair bandwidth required, for a given amount of redundancy. Further, we also note that design of codes, from the perspective of efficient recovery of its information elements for {error-correcting} (rather than erasure) erasure has also been studied in literature in associated with \emph{locally decodable codes} (See \cite{Yekhanin_Locally_Decodable} and references therein). The focus of this paper, however, is on MDS erasure codes (also referred to as minimum storage regenerating codes), i.e., $(n,k)$ codes which can tolerate any $(n-k)$ erasures. 

\section{Summary of Contributions}
The main contribution of this paper is the design a new class of MDS codes which achieve the minimum repair bandwidth of $\frac{n-1}{n-k}$ units for the repair of a single failed \emph{systematic} node. Our constructions operate with the new node downloading $\frac{1}{n-k}$ units from each of the $n-1$ surviving nodes for repair. The code constructions presented in this paper are listed below. \begin{enumerate}
\item \emph{Permutation Matrix Based Codes\footnote{The authors of reference \cite{Bruck_etal, Bruck_etal_ISIT} have discovered this class of permutation-matrix based codes in parallel work.} for General $(n,k)$: } In Section \ref{sec:Permutation}, we present a construction of codes which achieve the repair bandwidth lower bound of $\frac{n-1}{n-k}$ units for repair of systematic nodes for any tuple $(n,k)$ where $n>k$. The code generator submatrices of the construction are based on permutation matrices. The code construction, albeit finite, is based on random coding, with the random coding argument used to justify the existence of a repair-bandwidth optimal MDS code. This means that for any arbitrary $(n,k)$, a brute-force search over a (finite) set of codes described in the section, will yield a repair-bandwidth optimal MDS code. 
\item \emph{Explicit Construction for $n-k \in \{2,3\}$: } While Section \ref{sec:Permutation} describes a random coding based construction, we also provide in Section \ref{sec:explicit}, \emph{explicit} constructions for the special case of $n-k \in \{2,3\}$. 
\item \emph{Subspace Interference Alignment Framework for Optimal Repair:} In Section \ref{sec:Tensors}, we connect the idea of interference alignment via tensor (Kronecker) products, originally introduced in \cite{Suh_Tse_subspace}, to the Permutation matrix based codes developed in Section \ref{sec:Permutation}. The tensor-product based alignment framework, also termed subspace alignment in \cite{Suh_Tse_subspace}, provides a generalization of the Permutation matrix based codes developed in Section \ref{sec:Permutation}, and leads to a development of a family of MDS codes with optimal repair bandwidth. 
\end{enumerate}
It must be noted that the search for codes with efficiently repair both systematic and parity nodes is still open. However, from a practical perspective, the step taken in this paper is important since, in most storage systems, the number of parity nodes is small compared to systematic nodes.

\subsection{Efficient Code Construction in terms of Disk Access}
\label{sec:access}
While most previous works described above explore the repair problem by accounting for the amount of information to be sent over the network for repair, there exists another important cost during the repair of a node viz. amount of disk access. To understand the difference between these two costs, consider a toy example of a case where a disk stores two bits $a_1, a_2$. Now, suppose that, to repair some other failed node in this system, the bit $a_1+a_2$ has to be sent to a new node. This means that the bandwidth required for this particular disk is $1$ bit. However, in many storage systems, the disk-read speed is slower than the network transfer speeds and hence becomes a bottleneck. In the case where the disk read speed is a bottleneck, the defining factor in the speed of repair is the \emph{amount of disk access} rather than the repair bandwidth. In the toy example described the amount of disk access is $2$ bits as both $a_1$ and $a_2$ have to be read from the disk to compute $a_1+a_2$. Thus, it is possible that certain codes, while minimizing repair bandwidth, can perform poorly in terms of disk access rendering the codes impractical. In this paper, we will formalize this notion of disk access cost, and show that the codes based on permutation matrices in Section \ref{sec:Permutation} are not only bandwidth optimal, but also disk-access optimal, for the repair of a single failed systematic node.

\section{A Linear Algebraic Problem}
\label{sec:problem}
We begin by describing a linear algebraic problem which lies at the core of repair-optimal MDS codes. In particular, the problem described here is the problem we solve to find the optimal repair of $(n=k+2,k)$ codes. We start with a simpler problem which lies at the core of the special case where $(n=4,k=2)$ repair-optimal code and later generalize the problem.
\subsection*{Problem 1: A Simple Feasibility Problem} 
 Consider the following set of equations. 
\begin{eqnarray}
\mbox{rowspan} (\mathbf{V}_{1} \mathbf{H}_{2}) &=& \mbox{rowspan}({\mathbf{V}_1})\label{eq:1}\\
\mbox{rank}\left[\begin{array}{cc}\mathbf{V}_{1} \\ \mathbf{V}_1\mathbf{H}_{1}\end{array}\right] &=& L\label{eq:2}\\
\mbox{rowspan} \left(\mathbf{V}_{2} \mathbf{H}_{1}\right) &=& \mbox{rowspan}({\mathbf{V}_2})\label{eq:3}\\
\mbox{rank}\left[\begin{array}{c}\mathbf{V}_{2} \\ \mathbf{V}_{2}\mathbf{H}_{2}\end{array}\right] &=& L\label{eq:4}\\
\mbox{rank}(\mathbf{V}_{i})= \mbox{rank}(\mathbf{H}_{i})/2 &=& L/2, i=1,2 \label{eq:5}
\end{eqnarray}
where $\mathbf{H}_{1}$ and $\mathbf{H}_{2}$ be $L \times L$ matrices over some finite field. Now, the question of interest is, \emph{are the above set of equations feasible?} In other words, can we choose matrices $\mathbf{H}_{i},\mathbf{V}_{i}$ so that the above equations are satisfied. We assume that the field size and the size of $\mathbf{H}_{i},$ i.e., $L$ are parameters of choice. Because of (\ref{eq:5}), we can assume without loss of generality that $\mathbf{V}_{i},i=1,2$ are $L/2 \times L$ matrices.

Now, (\ref{eq:1}), (\ref{eq:3}) imply that the space spanned by the rows of $\mathbf{V}_{i}$ is an \emph{invariant} subspace of $\mathbf{H}_{j},$ for $i=1,2, j \in \{1,2\}-\{i\}$. Further, (\ref{eq:2}),(\ref{eq:4}) imply that none of the row vectors $\mathbf{V}_{i}$ lie in the span of $\mathbf{V}_{i}\mathbf{H}_{i}$ for $i=1,2$.\footnote{For the reader familiar with interference alignment literature in wireless communications, equations (\ref{eq:1}),(\ref{eq:3}) are similar to the conditions that all the interference align along $\mathbf{V}_{i}$, where $\mathbf{H}_{j},j\neq i$ is analogous to the channel matrix corresponding to an interfering link (See \cite{Cadambe_Jafar_int} for example). Similarly, conditions (\ref{eq:2}) and (\ref{eq:4}) are analogous to the condition that the desired signal appearing along matrix $\mathbf{H}_{i}$ is linearly resolvable from the aligned interference $\mathbf{V}_{i}$. The key difference between this problem and from most of interference alignment literature in wireless communications, is that, here, unlike in the latter, matrices $\mathbf{H}_{i}$ are design choices.} Before solving this problem, it is worth noting that $\mathbf{V}_i$ has to have at least $L/2$ linearly independent row vectors - or equivalently, a rank of at least $L/2$ - in order to satisfy (\ref{eq:2}),(\ref{eq:4}). Further, also note that, if we had allowed $\mathbf{V}_i, i=1,2$ to each have a rank as large as $L$ rather than $L/2$ in equation (\ref{eq:5}), the solution could have been trivial since any full rank matrices $\mathbf{V}_i, \mathbf{H}_{i},i=1,2$ would used to satisfy the conditions (\ref{eq:1})-(\ref{eq:4}). The question posed here, however, is whether there exist matrices $\mathbf{V}_{i}$ having exactly $L/2$ linearly independent row vectors, satisfying the above equations. It turns out that this problem has a fairly simple solution with $L=2$ and field size $q=5$. To see this, note that with $L=2$, (\ref{eq:1}) and (\ref{eq:3}) can be interpreted as eigen vector equations. Therefore, we can choose $\mathbf{V}_{1}^T$ to be an eigen vector of $\mathbf{H}_{2}^T$ and $\mathbf{V}_{2}^T$ to be an eigen vector of $\mathbf{H}_{1}^T$. As long as $\mathbf{H}_{1}^T$ and $\mathbf{H}_{2}^T$ can be chosen so that they have distinct (non-collinear) sets of eigen-vectors, the equations (\ref{eq:2}) and (\ref{eq:4}) are satisfied. It can be verified that in a field of size $5$, $\mathbf{H}_{1}$ and $\mathbf{H}_{2}$ can be chosen so that this property is satisfied. In fact, in a sufficiently large field size, the entries of $\mathbf{H}_{i},i=1,2$ can be randomly chosen independently, and uniformly over the entries of the field. With such a choice, it can be shown that, if $\mathbf{V}_{i}^T$ is chosen to be the eigen-vector of $\mathbf{H}_{j}^{T},j \neq i$ the equations (\ref{eq:1})-(\ref{eq:5}) are satisfied with a non-zero probability, thus guaranteeing feasibility.  The solution to this problem automatically implies that for $n=4,k=2,$ a single failed systematic node can be repaired by downloading exactly half the data stored in every surviving node (see Fig. \ref{fig:4,2}).

\begin{figure}
\center
\centerline{\psfig{figure=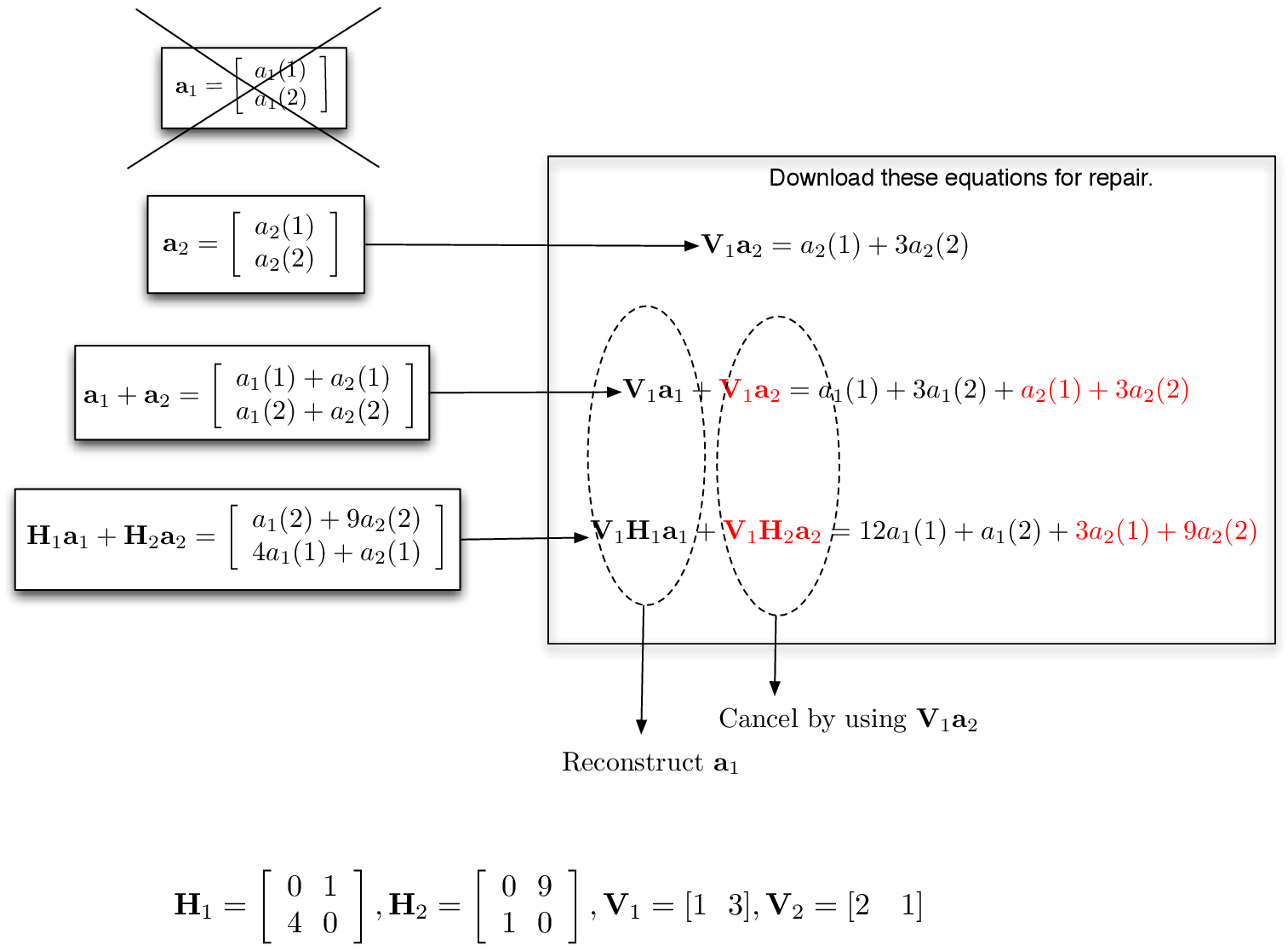, height=4 in}}
\caption{Repair of the first failed node in a (4,2) MDS code-based system along the lines of \cite{Wu_Dimakis}. Note that the repair is possible because (\ref{eq:1}) enables cancellation of $\mathbf{V}_{1}\mathbf{a}_{2}$ and (\ref{eq:2}) enables reconstruction of $\mathbf{a}_{1}$. Similarly equations (\ref{eq:3}),(\ref{eq:4}) enable repair when the second node fails, where $\mathbf{V}_{2}$ is used for obtaining linear combinations from the surviving nodes.}
\label{fig:4,2}
\end{figure}

\subsection*{Problem 2: Increase the number of constraints in Problem 1}
Now, let us generalize Problem 1. The goal of this generalized version is to verify the feasibility of the following equations, where  $\mathbf{H}_{i},i=1,2,\ldots,N$ are $L \times L$ matrices and $\mathbf{V}_{i},i=1,2,\ldots,N$ are  $L/2 \times L$ matrices.
\begin{eqnarray}
\mbox{span} (\mathbf{V}_{i} \mathbf{H}_{j}) &=& \mbox{span}({\mathbf{V}_i}), j\in\{1,2,\ldots,N\}-\{i\}\label{eq:1a}\\
\mbox{rank}\left[\begin{array}{cc}\mathbf{V}_{i} \\ \mathbf{H}_i\mathbf{V}_{i}\end{array}\right] &=& L, i=1,2,\ldots,N\label{eq:2a}\\
\mbox{rank}(\mathbf{V}_{i}) = \mbox{rank}(\mathbf{H}_{i})/2 &=& L/2, i=1,2,\ldots,N \label{eq:3a}
\end{eqnarray}
where $L$ is a parameter of choice. Here, it is worth noting two things. First, as before, if we had intended to find $L \times L$ matrices $\mathbf{V}_{i}$ satisfying (\ref{eq:1a}),(\ref{eq:2a}), the problem would have been trivial. Also, $\mathbf{V}_{i}$ can have no smaller than $L/2$ rows because of (\ref{eq:2a}). The question here, as before, is to construct $\mathbf{V}_{i}$ each of which have exactly $L/2$ row vectors satisfying the above conditions. The second point worth noting is that Problem $2$ is more challenging than Problem $1$ because the constraints here are more strict than the constraints of Problem 1. Problem $1$ is, in fact, a special case of the above problem when $N=2$. However, as $N$ increases, the number of constraints increases. This poses some additional constraints on the choice of matrices as compared to Problem 1. For instance, we will need $\mathbf{H}_{1}$ and $\mathbf{H}_{2}$ to have $N-2$ distinct common invariant subspaces $\mathbf{V}_{m},m\neq i, m\neq j$, in addition, to the condition that $\mathbf{V}_{1}$ (resp. $\mathbf{V}_{2}$) is invariant w.r.t. $\mathbf{H}_{2}$ (resp. $\mathbf{H}_{1}$) but linearly independent of $\mathbf{V}_{1}\mathbf{H}_{1}$ (resp. $\mathbf{H}_{2}$). Therefore, it is not clear at first sight whether the issue of feasibility can be resolved for arbitrary $N$. 

References \cite{Cadambe_Jafar_Maleki, Suh_Ramachandran} show that the above constraints can be satisfied asymptotically, as $L \to \infty,$ by using random diagonal matrices for $\mathbf{H}_{i}$ and the asymptotic interference alignment solution of \cite{Cadambe_Jafar_int} to construct $\mathbf{V}_{i}$ for $i=1,2,\ldots,N$. However, it was not known whether the above set of constraints is feasible when $L$ is restricted to be finite - it is this open problem that is solved in this paper. In particular, we will use a tensor-product based framework which enables us to decompose this problem into several instances of Problem 1 and hence show feasibility. Put differently, the framework will enable us to stitch multiple instances of problem $1$ using the idea of tensor products to solve the above problem.

The rest of this paper is organized as follows. We will first present a framework used for our repair optimal code in the next section. In this next section, we will also connect the repair problem to the problem presented above. In Sections \ref{sec:Permutation} and \ref{sec:explicit}, we will respectively present random codes and explicit codes based on permutation matrices which are optimal from the perspective of repair of a single systematic node. These constructions can be interpreted as a solution to the above problem where $\mathbf{H}_{i}$ are permutation matrices. In Section \ref{sec:Tensors}, we will revisit the problem described above, and present our tensor-product based framework to solve this problem. The framework of Section \ref{sec:Tensors} generalizes the permutation matrix based construction of \ref{sec:Permutation}. 

\section{System Model -  Optimal Repair for an $(n,k)$ MDS Code}
In this section we present a general framework for optimal repair a single failed node in a linear MDS code based distributed storage system.
Consider $k$ sources, all of equal size $\LL=\MM/k$ over a field $\mathbb{F}_{q}$ of size $q$. Source $i\in\{1,2,\ldots,k\}$ is represented by the $\LL \times 1$ vector $\mathbf{a}_i \in \mathbb{F}_q^{\LL}$. Note here that $\MM$ denotes the size of the total information stored in the distributed storage system, in terms of the number of elements over the field. There are $n$ nodes storing the $k$ source (vector) symbols using an $(n,k)$ MDS code. Each node stores a data of size $\LL$, i.e., each coded (vector) symbol of the $(n,k)$ code is a $\LL \times 1$ vector. Therefore, $1$ unit is equivalent to $\LL$ scalars over the field $q$. The data stored in node $i$ represented by $\LL \times 1$ vector $\mathbf{d}_i$, where $i=1,2,\ldots,n$.
We assume that our code is linear and $\mathbf{d}_i$ can be represented as
$$\mathbf{d}_i = \sum_{j=1}^{k} \mathbf{C}_{i,j} \mathbf{a}_j,$$
where $\mathbf{C}_{i,j}$ are $\LL \times \LL$ square matrices.
Further, we restrict our codes to have a systematic structure, so that, for $i \in \{1,2,\ldots,k\}$,
$$\mathbf{C}_{i,j} = \left\{ \begin{array}{cc} \mathbf{I} & j = i\\ \mathbf{0} & j \neq i \end{array} \right\}.$$ Since we restrict our attention to MDS codes, we will need the matrices $\mathbf{C}_{i,j}$ to satisfy the following property
\begin{property}
\label{property:MDS}
\begin{equation}\mbox{rank}\left(\left[ \begin{array}{cccc} \mathbf{C}_{j_1, 1}& \mathbf{C}_{j_1, 2}& \ldots &\mathbf{C}_{j_1, k}\\
\mathbf{C}_{j_2, 1}& \mathbf{C}_{j_2, 2}& \ldots & \mathbf{C}_{j_2, k}\\
\vdots& \vdots & \ddots & \vdots\\
\mathbf{C}_{j_k, 1}& \mathbf{C}_{j_k, 2}& \ldots & \mathbf{C}_{j_k, k} \end{array} \right] \right) = \LL k = \MM \label{eq:MDS} \end{equation}
for any distinct $j_1, j_2, \ldots, j_k \in \{1,2,\ldots, n\}$.
\end{property}
The MDS property ensures that the storage system can tolerate up to $(n-k)$ failures (erasures), since all the sources can be reconstructed from any $k$ nodes whose indices are represented by $j_1,j_2,\ldots,j_k \in \{1,2,\ldots,n\}$. 
Now, consider the case where a single systematic node, say node $i \in \{1,2,\ldots,k\}$ fails. The goal here is to reconstruct the failed node $i$, i.e., to reconstruct $\mathbf{d}_i,$ using all the other $n-1$ nodes, i.e., $\{ \mathbf{d}_{j}: j \neq i\}$. 
To understand the solution, first, consider the case where node $1$ fails. We download a fraction of $\frac{1}{n-k}$ of the data stored in each of the nodes $\{1,2,3,\ldots,n\} - \{1\},$ so that the total repair bandwidth is $\frac{n-1}{n-k}$ units. We focus on linear repair solutions for our codes, which implies that we need to download $\frac{\LL}{n-k}$ linear combinations from each of $\mathbf{d}_{j}, j \in \{2,3,\ldots,n\}$. Specifically, we denote the linear combination downloaded from node $j\in\{2,3,\ldots,n\}$ as 
\begin{eqnarray*} \mathbf{V}_{1,j}\mathbf{d}_{j}&=&\mathbf{V}_{1,j} \sum_{i=1}^{k} \mathbf{C}_{j,i} \mathbf{a}_i\\
&=& \underbrace{\mathbf{V}_{1,j} \mathbf{C}_{j,1} \mathbf{a}_1}_{\mbox{Desired signal component}}+ \underbrace{\mathbf{V}_{1,j} \sum_{i=2}^{k} \mathbf{C}_{j,i} \mathbf{a}_i}_{\mbox{Interference component}},
\end{eqnarray*}
where $\mathbf{V}_{1,j}$ is a $\frac{\LL}{n-k} \times \LL$ dimensional matrix. The matrices $\mathbf{V}_{1,j}$ are referred to as \emph{repair matrices} in this paper. The goal of the problem is to construct $\LL$ components of $ \mathbf{a}_{1}$ from the above equations. For systematic node $j \in\{2,3,\ldots, k\}$, the equations downloaded by the new node do not contain information of the desired signal $\mathbf{a}_1$, since for these nodes, $\mathbf{C}_{j,1}=\mathbf{0}$. The linear combinations downloaded from the remaining nodes $j \in \{k+1,k+2,\ldots,n\}$, however, contain components of both the desired signal and the interference. Thus, the downloaded linear combinations $\mathbf{V}_{1,j} \mathbf{d}_{j}$ are of two types.
\begin{enumerate}
\item The data downloaded from the surviving systematic nodes $i=2,\ldots, k$ contain no information of the desired signal $\mathbf{a}_1$, i.e., 
$$\mathbf{V}_{1,j} \mathbf{d}_{j} = \mathbf{V}_{1,j} \mathbf{a}_{j}, j=2,\ldots, k.$$
Note that there $\frac{\LL}{n-k}$ such linear combinations of each interfering component $\mathbf{a}_{j},j=2,3,\ldots,k$.
\item Now, from each of the $n-k$ parity nodes, $\frac{\LL}{n-k}$ linear combinations are downloaded. Therefore, a total of $\LL$ linear combinations are downloaded from parity nodes. The $\LL$ components of the desired signal have to be reconstructed using these $\LL$ linear combinations of the form $\mathbf{V}_{1,j}\mathbf{d}_{j}, j=k+1, k+2,\ldots,n$. Note here that these are $\mathcal{L}$ linear equations in $k\mathcal{L}$ scalars - the $\mathcal{L}$ desired components of $\mathbf{a}_{1}$ \emph{and} $(k-1)\mathcal{L}$ interfering components of $\mathbf{a}_{2},\mathbf{a}_{3},\ldots, \mathbf{a}_{k}.$  For successful reconstruction of the desired signal, the interference terms associated with $\mathbf{a}_{j}, j=2\ldots,k$ contained in these linear combinations need to be cancelled completely.
\end{enumerate}
 The goal of our solution will be to completely cancel the interference from the second set of $\LL$ linear combinations, using the first set of linear combinations. Then $\mathbf{a}_1$ is regenerated using this second set of $\LL$ interference-free linear combinations (See Fig. \ref{fig:repair53}).
\subsection{Interference Cancellation}
\begin{figure}
\center
\centerline{\psfig{figure=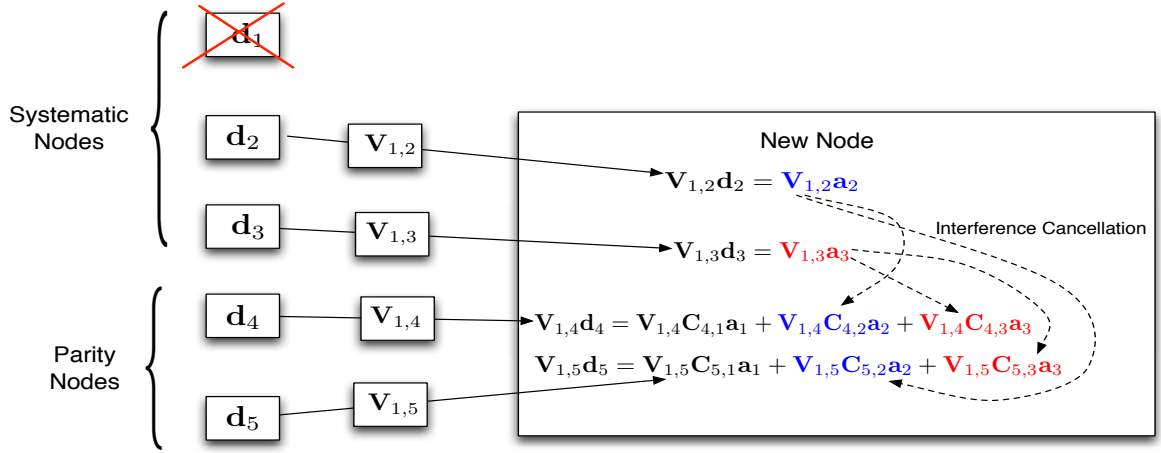,width=6.1 in,height=2.4in}}
\caption{Repair of first node for $n=5,k=3$. Equation (\ref{eq:alignmenta}) ensures that interference cancellation is possible}
\label{fig:repair53}
\end{figure}

 The linear combinations corresponding to interference component $\mathbf{a}_{i}, i \neq 1$ downloaded using node $i$ by the new node is $\mathbf{V}_{1,i} \mathbf{a}_i$ for $i=2,3,\ldots,k$. To cancel the associated interference from all the remaining nodes $\mathbf{V}_{1,j} \mathbf{d}_{j}$  by linear techniques, we will need, $\forall j=k+1, k+2, \ldots, n$, $\forall i=2,3,\ldots,k$
\begin{eqnarray}
\mbox{rowspan}(\mathbf{V}_{1,j} \mathbf{C}_{j,i}) &\subseteq& \mbox{rowspan}(\mathbf{V}_{1,i}), \nonumber\\
\Rightarrow \mbox{rowspan}(\mathbf{V}_{1,j} \mathbf{C}_{j,i}) &=& \mbox{rowspan}(\mathbf{V}_{1,i}), \label{eq:alignmenta} \label{eq:alignment}\end{eqnarray}
where (\ref{eq:alignmenta}) follows because $\mathbf{C}_{j,i}$ are all full rank matrices and therefore, the subset relation automatically implies the equality relation as $\mbox{rank}(\mathbf{V}_{1,j} \mathbf{C}_{j,i}) = \mbox{rank}(\mathbf{V}_{1,j})=\frac{\LL}{n-k} = \mbox{rank}(\mathbf{V}_{1,i})$. Thus, as long as (\ref{eq:alignmenta}) is satisfied for all values $j \in \{k+1,k+2,\ldots,n\},i \in \{2,3,\ldots,k\}$, the interference components can be completely cancelled from $\mathbf{V}_{1,j} \mathbf{d}_{j}$ to obtain $\mathbf{V}_{1,j} \mathbf{C}_{j,1} \mathbf{a}_{1}, j \in \{k+1,k+2,\ldots,n\}$ (See Fig. \ref{fig:repair53}).
 Now, we need to ensure that the desired $\LL \times 1$ vector $\mathbf{a}_1$ can be uniquely resolved from the $\LL$ linear combinations of the form $\mathbf{V}_{1,j} \mathbf{C}_{j,1} \mathbf{a}_1, j=k+1,k+2,\ldots,n$. In other words, we need to ensure that 
\begin{eqnarray}
\mbox{rank}\left(\left[\begin{array}{c}\mathbf{V}_{1,k+1} \mathbf{C}_{k+1,1} \\ \mathbf{V}_{1,k+2} \mathbf{C}_{k+2,1} \\ \vdots \\ \mathbf{V}_{1,n} \mathbf{C}_{n,1}  \end{array} \right] \right) &=& \LL
\label{eq:reconstruction0}
\end{eqnarray}
If we construct $\mathbf{C}_{l,j}$ and $\mathbf{V}_{1,i}$  satisfying (\ref{eq:alignmenta}) and (\ref{eq:reconstruction0}) for $i=2,\ldots, n, j=1,2,\ldots, k, l=1,2,3\ldots, n$, then, a failure of node $1$ can be repaired with the desired minimum repair bandwidth. To solve the problem for the failure of any other systematic node, we need to ensure similar conditions. We summarize all the conditions required for successful reconstruction of a single failed (systematic) node with the minimum repair bandwidth below.
\begin{itemize}
\item Equation (\ref{eq:MDS}) in Property \ref{property:MDS}.
\item The interference alignment relations.
\begin{eqnarray} \mbox{rowspan}(\mathbf{V}_{l,j}\mathbf{C}_{j,i}) = \mbox{rowspan}(\mathbf{V}_{l,i})\label{eq:alignment0} \end{eqnarray}
for $l=1,2,\ldots,k$, $j=k+1,k+2,\ldots,n$ and $i \in \{1,2,\ldots,k\} - \{l\}$
\item Reconstruction of the failed node, given that the alignment relations are satisfied.
\begin{eqnarray}
\mbox{rank}\left(\left[\begin{array}{c}\mathbf{V}_{l,k+1} \mathbf{C}_{k+1,l} \\ \mathbf{V}_{l,k+2} \mathbf{C}_{k+2,l} \\ \vdots \\ \mathbf{V}_{l,n} \mathbf{C}_{n,l}  \end{array} \right] \right) &=& \LL
\label{eq:reconstruction}
\end{eqnarray}
for $l=1,2,\ldots,k$.
\end{itemize}
Note that given $n,k$, our design choices are $\LL$, $q$, $\mathbf{C}_{j,i}$ and $\mathbf{V}_{l,j}$ for $l=1,2,\ldots,k$, $j=k+1,k+2,\ldots,n$ and $i \in \{1,2,\ldots,k\} - \{l\}.$

Reference \cite{Shah_etal} has shown that the above conditions \emph{cannot} be satisfied if we restrict ourselves to $\MM= k(n-k)$. References \cite{Cadambe_Jafar_Maleki, Suh_Ramachandran_general} constructed solutions which satisfied the above relations in an asymptotically exact, as $\MM \to \infty$. The main contribution of this paper is the construction of coding sub-matrices and repair matrices so that the above relations are satisfied exactly, with finite $\MM$, i.e., with $\MM=k (n-k)^k$. 

\subsection{Connections to Problem 2 in Section \ref{sec:problem}}
\label{sec:connections}
Above, we have defined a general structure for a linear, repair bandwidth optimal solution. In the specific solution described in this paper, the repair matrices satisfy an additional property: in our solution, $$ \mathbf{V}_{l,j} = \mathbf{V}_{l,j^{'}}$$ for all $l \in \{1,2,\ldots,k\}, j\neq j^{'}, j,j^{'} \in \{1,2,\ldots,n\}-\{l\}$. In other words, when a node, say node $l$, fails, we download the \emph{same} linear combination from every surviving node. We use the notation $$ \mathbf{V}_{l} \define \mathbf{V}_{l,j}$$ for all $j \in \{1,2,\ldots,n\}-\{l\}$. Further, in our solution the coding sub-matrices associated with the first parity node are all (scaled) identity matrices, i.e., $\mathbf{C}_{k+1,i} = \lambda_{k+1,i}\mathbf{I}_{\LL}$ for $i=1,2,\ldots,k$ where $\lambda_{k+1,i}$ is a scalar over the field $\mathbb{F}_{q}$,
so that $$\mathbf{d}_{k+1} = \sum_{i=1}^{k}\lambda_{k+1,i} \mathbf{a}_{i}$$
Now, with these choices, it can noted for $n-k=2,$ equations (\ref{eq:alignment0}) and (\ref{eq:reconstruction}) are equivalent to (\ref{eq:1a}) and (\ref{eq:2a}) in the previous section, where $k=N, \mathcal{L}=L$ and $\mathbf{C}_{k+2,i}=\mathbf{H}_{i}.$ Thus, the problems motivated in the previous section lie at the core of the repair problem.

\subsection{Disk-Access Optimality}
Our solution satisfies a disk-access optimality property which is defined formally here. 
\begin{definition}
Consider a set of $\LL \times \LL$ dimensional coding sub-matrices  $\mathbf{C}_{i,j}, i=k+1,k+2,\ldots, n, j=1,2,\ldots,k$ and a set of repair matrices $\mathbf{V}_{l,i}$ for some  $l\in\{1,2,\ldots,n\}$ and for all $i\in\{1,2,\ldots,n\}-\{l\},$ where the repair matrix $\mathbf{V}_{l,i}$ has dimension $B_{l,i} \times \LL,$ where $B_{l,i}\leq \LL$.  The repair matrices satisfy the property that $\mathbf{d}_{l}$ can be reconstructed linearly from $\mathbf{V}_{l,i} \mathbf{d}_{i}, i \in \{1,2,\ldots,n\}-\{l\}$. In other words, a failure of node $l$ can be repaired using the repair matrices. Then the amount of \emph{disk access} required for the repair of node $l$ is defined to be the quantity $$ \sum_{i=\{1,2,\ldots,n\}-\{l\}} \omega(\mathbf{V}_{l,i}) $$ 
where $\omega(\mathbf{A})$ represents the number of non-zero columns of matrix $\mathbf{A}$.
\end{definition}
To compute $\mathbf{V}_{l,i} \mathbf{d}_{i}$, only $\omega(\mathbf{V}_{l,i})$ entries of the matrix $\mathbf{d}_{i}$ have to be accessed. This leads to the above definition for the amount of disk access for a linear solution. Also, note that if $\textrm{rank}(\mathbf{V}_{l,i})$ - the amount of bandwidth used - is always smaller than $\omega(\mathbf{V}_{l,i}).$ Therefore, the amount of disk access is smaller than the amount of bandwidth used for a given solution. This leads to the following lemma.
\begin{lemma}
For any $(n,k)$ MDS code storing $1$ unit of data in each disk, the amount of disk access needed to repair any single failed node $l=1,2,\ldots,n$ is at least as large as $\frac{n-1}{n-k}$ units.
\end{lemma}

Our code constructions based on permutation matrices presented in the next section are not only repair bandwidth optimal, but it are also optimal in terms of disk access since they meets the bound of the above lemma. More formally, for our solution $\mathbf{V}_{j}$ not only has a rank of $\LL/(n-k)$, it also has exactly $\LL/(n-k)$ non-zero columns; in fact, $\mathbf{V}_{j}$ has exactly $\LL/(n-k)$ non-zero entries. Among the $\LL$ columns of $\mathbf{V}_{j}$, $\LL-\frac{\LL}{n-k}$ columns are zero. This means that, to obtain the linear combination $\mathbf{V}_{l} \mathbf{d}_{i}$ from node $i$ for repair of node $l \neq i$, only $\frac{\LL}{n-k}$ entries of the node $i$ has to be accessed. We now proceed to describe our solution.  


\section{Optimal Codes via Permutation Matrices}
\label{sec:Permutation}
In this section, we describe a set of random codes based on permutation matrices satisfying the desired properties described in the previous section. We begin with some preliminary notations required for our description.
\subsubsection*{Notations and Preliminary Definitions}
The bold font is used for vectors and matrices and the regular font is reserved for scalars. Given a $l \times 1$ dimensional vector $\mathbf{a}$ its $l$ components are denoted by 
$$ \mathbf{a} = \left[\begin{array}{c}a(1)\\a(2)\\\vdots \\a(l)\end{array} \right]$$
For example, $\mathbf{d}_{1} = \left[d_1(1)~~d_1(2)~~\ldots~~d_1(\LL) \right]^T$.  Given a set $\mathcal{A}$, the $l$-dimensional Cartesian product of the set is denoted by $\mathcal{A}^{l}$. The notation $\mathbf{I}_{l}$ denotes the $l \times l$ identity matrix; the subscript $l$ is dropped when the size $l$ is clear from the context. Next, we define a set of functions which will be useful in the description of our codes.

Given $(n,k)$ and a number $m \in \{1,2,\ldots,(n-k)^{k}\},$ we define a function\footnote{While the functions defined here are parametrized by $n,k$, these quantities are not explicitly denoted here for brevity of notation} $\fbeta: \{1,2,\ldots,(n-k)^{k}\} \rightarrow \{0,1,\ldots,(n-k-1)\}^{k}$ such that ${\fbeta}(m)$ is the unique $k$ dimensional vector whose $k$ components represent the $k$-length representation of $m-1$ in base $(n-k)$. In other words 
$$ \fbeta(m) = (r_1,r_2,\ldots,r_k) \Leftrightarrow m-1 = \sum_{i=1}^{k} r_i (n-k)^{i-1},$$
where $r_i \in \{0,1,\ldots,(n-k-1)\}$. 
Further, we denote the $i$th component of $\fbeta(m)$ by $\gbeta_i(m),$ for $i=1,2,\ldots,k$. Since the $k$-length representation of a number in base $(n-k)$ is unique, $\fbeta$ and $\gbeta_{i}$ are well defined functions. Further, $\fbeta$ is invertible and its inverse is denoted by $\gbeta^{-1}$. We also use the following compressed notation for $\gbeta^{-1}$.
$$ \langle r_1,r_2,\ldots,r_k \rangle \define \gbeta^{-1}(r_1,r_2,\ldots,r_k) = \sum_{i=1}^{k} r_i (n-k)^{i-1} - 1$$

The definition of the above functions will be useful in constructing our codes.

\begin{figure}
\center
\centerline{\psfig{figure=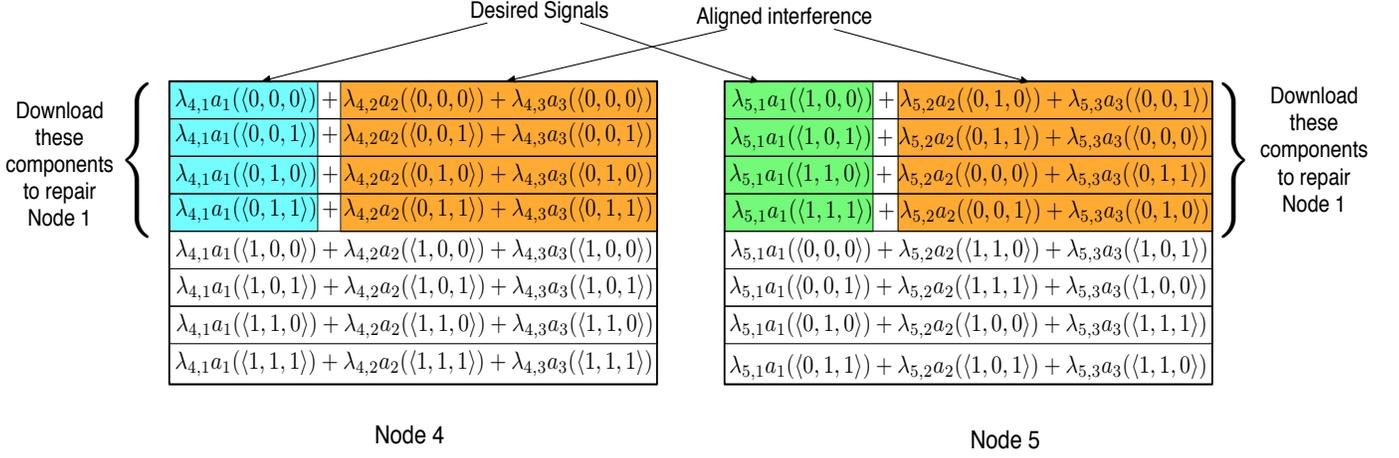,width=7.2 in,height=2.4in}}
\caption{The two parity nodes in the $(5,3)$ code and repair strategy for failure of node $1$. Shaded portions indicate downloaded portions used to recover failure of node 1. Note that the undesired symbols can be cancelled by downloading half the components of $\mathbf{a}_2, \mathbf{a}_3$, i.e., by downloading $a_2(\langle0,x_1,x_2\rangle)$ and $a_3(\langle0,x_1,x_2\rangle)$ for $x_1,x_2 \in \{0,1\}$.} 
\label{fig:permutations_alignment}
\end{figure}

\subsection{Example : n=5, k=3}
We motivate our code by first considering the case where $k=3,n=5$ for simplicity. The extension of the code to arbitrary $n,k$ will follow later\footnote{Optimal codes for $n=5,k=3$ have been proposed in \cite{Cullina_Dimakis_Ho, Suh_Ramachandran}. We only use this case to demonstrate our construction in the simplest non-trivial setting.} For $n=5,k=3$, we have $\MM/k = (n-k)^{k}=2^3=8$. As the name suggests, we use scaled permutation matrices for $\mathbf{C}_{i,j}, j \in \{1,2,\ldots,k\}, i \in \{k+1,k+2,\ldots,n\}$. Note here that the variables $\mathbf{a}_{j},j=1,2,\ldots,k$ are $(n-k)^{k} \times 1$ dimensional vectors. We represent the $(n-k)^{k}=8$ components these vectors by the $k=3$ bit representation of their indices as

$$\mathbf{a}_{j} = \left( a_j(1)~~a_j(2)~~\ldots~~a_j(8)\right)^T = \left( \begin{array}{c} a_j\left(\langle 0,0,0 \rangle\right) \\ a_j(\langle 0,0,1\rangle) \\ a_j(\langle 0,1,0\rangle) \\ a_j(\langle 0,1,1 \rangle) \\ a_j(\langle 1,0,0\rangle) \\ a_j(\langle 1,0,1\rangle) \\ a_j(\langle 1,1,0\rangle) \\ a_j(\langle 1,1,1\rangle) \end{array} \right)$$
for all $j=1,2,\ldots,k$.
Now, similarly, we can denote the identity matrix as 
$$ \mathbf{I}_{8} = \left[\begin{array}{c} \mathbf{e}(1)\\ \mathbf{e}(2)\\ \vdots \\ \mathbf{e}(8) \end{array}\right] = \left[\begin{array}{c} \mathbf{e}(\langle 0,0,0 \rangle) \\ \mathbf{e}(\langle 0,0,1 \rangle) \\ \vdots \\ \mathbf{e}(\langle 1,1,1 \rangle)\end{array}\right],$$ where, naturally, $\mathbf{e}(i)$ is the $i$th row of the identity matrix. 
Now, we describe our code as follows. Since the first three storage nodes are systematic nodes and the remaining two are parity nodes, the design parameters are $\mathbf{C}_{4,j}, \mathbf{C}_{5,j}, \mathbf{V}_{j}$ for $j=1,2,3$. We choose
$$ \mathbf{C}_{4,j}=\lambda_{4,j} \mathbf{I}$$ so that
$$ \mathbf{d}_{4} = \sum_{j=1}^{3} \lambda_{4,j}\mathbf{a}_{j},$$
where $\lambda_{4,j}$ are independent random scalars chosen using a uniform distribution over the field $\mathbb{F}_{q}$.
Now, consider the $8 \times 8$ permutation matrix $\mathbf{P}_{i}$ defined as 
$$ \mathbf{P}_{1} = \left[\begin{array}{c} \mathbf{e}(\langle 1,0,0\rangle) \\ \mathbf{e}(\langle1,0,1\rangle)\\ \mathbf{e}(\langle 1,1,0\rangle )\\ \mathbf{e}(\langle 1,1,1 \rangle )\\ \mathbf{e}(\langle 0,0,0 \rangle )\\  \mathbf{e}(\langle 0,0,1\rangle )\\ \mathbf{e}(\langle 0,1,0\rangle )\\ \mathbf{e}(\langle 0,1,1\rangle ) \end{array}\right],
\mathbf{P}_{2} = \left[ \begin{array}{c}\mathbf{e}(\langle 0,1,0\rangle )\\ \mathbf{e}(\langle 0,1,1\rangle )\\ \mathbf{e}(\langle 0,0,0\rangle )\\ \mathbf{e}(\langle 0,0,1\rangle)\\ \mathbf{e}(\langle 1,1,0\rangle )\\ \mathbf{e}(\langle 1,1,1\rangle )\\ \mathbf{e}(\langle 1,0,0\rangle )\\ \mathbf{e}(\langle 1,0,1\rangle )\end{array}\right],
\mathbf{P}_{3} = \left[ \begin{array}{c}\mathbf{e}(\langle 0,0,1\rangle )\\ \mathbf{e}(\langle 0,0,0\rangle )\\ \mathbf{e}(\langle 0,1,1\rangle)\\ \mathbf{e}(\langle 0,1,0\rangle )\\ \mathbf{e}(\langle 1,0,1\rangle )\\  \mathbf{e}(\langle 1,0,0\rangle )\\ \mathbf{e}(\langle 1,1,1\rangle )\\ \mathbf{e}(\langle 1,1,0\rangle )\end{array}\right]$$
Then, the fifth node (i.e., the second parity node) is designed as 
$$\mathbf{d}_{5} = \sum_{j=1}^{3} \lambda_{5,j} \mathbf{P}_{j}\mathbf{a}_{j},$$ where $\lambda_{5,j}$ are random independent scalars drawn uniformly over the entries of the field $\mathbb{F}_{q}$. In other words, we have 
$$\mathbf{C}_{5,j} = \lambda_{5,j} \mathbf{P}_j, j=1,2,3.$$
The code is depicted in Figure \ref{fig:permutations_alignment}. For a better understanding of the structure of the permutations, consider an arbitrary column vector  $\mathbf{a} = \left[a(1)~~a(2)~~\ldots~~a(8)\right]^T$. Then, $$\mathbf{P}_{1} \mathbf{a} = \left( \begin{array}{c} a(\langle 1,0,0\rangle) \\ a(\langle 1,0,1\rangle ) \\ a(\langle 1,1,0\rangle ) \\ a(\langle 1,1,1\rangle) \\a(\langle 0,0,0\rangle ) \\ a(\langle 0,0,1\rangle ) \\ a(\langle 0,1,0\rangle ) \\ a(\langle 0,1,1\rangle ) \end{array} \right) = \left( \begin{array}{c} a(5) \\ a(6) \\ a(7) \\ a(8) \\ a(1) \\ a(2) \\ a(3) \\ a(4) \end{array} \right)$$
In other words, $\mathbf{P}_{1}$ is a permutation of the components of $\mathbf{a}$ such that the element $a(\langle 1,x_2,x_3\rangle)$ is swapped with the element $a(\langle 0,x_2,x_3\rangle)$ for $x_2,x_3 \in \{0,1\}$. Similarly, $\mathbf{P}_{2}$ swaps $a(\langle x_1,0,x_3\rangle)$ with $a(\langle x_1,1,x_3\rangle)$ and $\mathbf{P}_{3}$ swaps $a(\langle x_1,x_2,0\rangle)$ with $a(\langle x_1,x_2,1\rangle)$ where $x_1,x_2,x_3 \in \{0,1\}$.


Now, we show that this code can be used to achieve optimal recovery, in terms of repair bandwidth, for a single failed systematic node. To see this, consider the case where node $1$ fails. Note that for optimal repair, the new node has to download a fraction of $\frac{1}{n-k}=\frac{1}{2}$ of every surviving node, i.e., nodes $2,3,4,5$. The repair strategy is to download $\mathbf{d}_{i}(\langle 0,0,0\rangle ), \mathbf{d}_{i}(\langle 0,0,1\rangle ),\mathbf{d}_{i}(\langle 0,1,0\rangle ),\mathbf{d}_{i}(\langle 0,1,1\rangle )$ from node $i \in\{2,3,4,5\}$, so that 
$$\mathbf{V}_{1} = \left[\begin{array}{c} \mathbf{e}(\langle 0,0,0\rangle ) \\ \mathbf{e}(\langle 0,0,1\rangle )\\ \mathbf{e}(\langle 0,1,0\rangle ) \\ \mathbf{e}(\langle 0,1,1\rangle )\end{array}\right] = \left[\begin{array}{c} \mathbf{e}(1) \\ \mathbf{e}(2)\\ \mathbf{e}(3) \\ \mathbf{e}(4)\end{array}\right] $$

In other words, the rows of $\mathbf{V}_{1}$ come from the set $\{\mathbf{e}(\langle 0,x_2,x_3\rangle) : x_2, x_3 \in\{0,1\}\}$. Note that the strategy downloads half the data stored in every surviving node as required. With these download vectors, it can be observed (See Figure \ref{fig:permutations_alignment}) that the interference is aligned as required and all the $8$ components of the desired signal $\mathbf{a}_{1}$ can be reconstructed. Specifically we note that 
\begin{equation} \mbox{rowspan}(\mathbf{V}_{1} \mathbf{C}_{4,i}) = \mbox{rowspan}(\mathbf{V}_{1} \mathbf{C}_{5,i}) = \mbox{span}(\{\mathbf{e}(\langle 0,x_2,x_3\rangle): x_2,x_3\in\{0,1\} \})\label{eq:alignment2} \end{equation}
for $i=2,3$:
Put differently, because of the structure of the permutations, the downloaded components can be expressed as 
$${d}_{4}(\langle 0,x_2,x_3\rangle) = \lambda_{4,1}a_1(\langle0,x_2,x_3\rangle)+ \lambda_{4,2}a_2(\langle0,x_2,x_3\rangle)+\lambda_{4,3}a_3(\langle0,x_2,x_3\rangle)  $$
$${d}_{5}(\langle 0,x_2,x_3\rangle) = \lambda_{5,1}a_1(\langle1,x_2,x_3\rangle)+ \lambda_{5,2}a_2(\langle0,x_2\oplus 1,x_3\rangle)+\lambda_{5,3}a_3(\langle0,x_2,x_3\oplus 1\rangle)  $$
Note that since $x_2,x_3 \in \{0,1\}$ there are a total $8$ components described in the two equations above, such that, all the interference is of the form $a_i(\langle 0,y_2,y_3 \rangle), i \in \{2,3\}, y_2,y_3 \in \{0,1\}$. In other words, the interference from $\mathbf{a}_{i},i=2,3$ comes from only half its components, and the interference is aligned as described in (\ref{eq:alignment2}). However, note that the $8$ components span all the $8$ components of the desired signal $\mathbf{a}_{1}$.
Thus, the interference can be completely cancelled and the desired signal can be completely reconstructed.
 
 Similarly, in case of failure of node $2$, the set of rows of the repair matrices $\mathbf{V}_{2}$ is equal to the set $\{\mathbf{e}(\langle x_1,0,x_3 \rangle) : x_1, x_3 \in \{0,1\} \}$, i.e., 
$$\mathbf{V}_{2} = \left[\begin{array}{c} \mathbf{e}(\langle 0,0,0\rangle ) \\ \mathbf{e}(\langle 0,0,1\rangle )\\ \mathbf{e}(\langle 1,0,0\rangle ) \\ \mathbf{e}(\langle 1,0,1\rangle )\end{array}\right] = \left[\begin{array}{c} \mathbf{e}(1) \\ \mathbf{e}(2)\\ \mathbf{e}(5) \\ \mathbf{e}(6)\end{array}\right] $$ 
With this set of download vectors, it can be noted that, for $i=1,3$
\begin{equation} \mbox{rowspan}(\mathbf{V}_{2} \mathbf{C}_{4,i}) = \mbox{rowspan}(\mathbf{V}_{2} \mathbf{C}_{5,i} ) = \mbox{span}(\{\mathbf{e}(\langle x_1,0,x_3\rangle): x_1,x_3\in\{0,1\} \})\label{eq:alignment1} \end{equation}
so that the interference is aligned. It can be verified that the desired signal can be reconstructed completely because of condition (\ref{eq:reconstruction}) as well.
 The rows of $\mathbf{V}_{3}$ come from the set $\{\mathbf{e}(\langle x_1, x_2, 0\rangle) : x_1, x_2 \in \{0,1\} \}.$ Equations (\ref{eq:alignment}) and (\ref{eq:reconstruction}) can be verified to be satisfied for this choice of $\mathbf{V}_{3}$ with the alignment condition taking the form this case can be verified to be satisfied, for $i=1,2$, as 
\begin{equation} \mbox{rowspan}(\mathbf{V}_{3} \mathbf{C}_{4,i}) = \mbox{rowspan}( \mathbf{V}_{3} \mathbf{C}_{5,i}) = \mbox{span}(\{\mathbf{e}(\langle x_1,x_2,0\rangle): x_1,x_2\in\{0,1\} \})\label{eq:alignment3} \end{equation} 

 While this shows that optimal repair is achieved, all the remains to be shown is that the code is an MDS code, i.e., Property \ref{property:MDS}. This is shown in Appendix \ref{app:MDS}, for the generalization of this code to arbitrary values of $(n,k)$. Next, we describe this generalization. 


\subsection{The optimal $(n,k)$ code}
This is a natural generalization of the $(5,3)$ code for general values of $(n,k)$, with $\LL = (n-k)^{k}$. 
To describe this generalization, we define function $\lbeta_{i}(m) = (\gbeta_{1}(m),\gbeta_{2}(m),,\ldots,\gbeta_{i-1}(m),\gbeta_{i}(m)\oplus 1, \gbeta_{i+1}(m), \gbeta_{i+2}(m),\ldots,\gbeta_{k}(m))$, where the operator $\oplus$ represents an addition modulo $(n-k)$.
In other words, $\lbeta_{i}(m)$ essentially modifies the $i$th position in the base $(n-k)$ representation of $m-1$, by addition of $1$ modulo $(n-k)$. 
\begin{remark} For the optimal $(5,3)$ code described previously, note that the $m$th row of $\mathbf{P}_{i}$ is $\mathbf{e}(\langle \lbeta_{i}(m) \rangle)$. In other words, for the $(5,3)$ code described above, the $m$th component of $\mathbf{P}_{i} \mathbf{a}$ is equal to $a(\langle \lbeta_{i}(m)\rangle)$.
\label{remark1}
\end{remark}
\begin{remark}
 $\langle \lbeta_i(1)\rangle, \langle \lbeta_i(2)\rangle, \ldots, \langle\lbeta_i((n-k)^{k})\rangle$ is a permutation of $1,2,\ldots,(n-k)^{k}$ for any $i \in \{1,2,\ldots,k\}$. Therefore, given a $\LL \times 1$ vector $\mathbf{a}$,  
$$\left[a\big(\langle \lbeta_i(1)\rangle\big), a\big(\langle \lbeta_i(2)\rangle\big), \ldots, a\big(\langle\lbeta_i((n-k)^{k})\rangle \big)\right]^T$$ is a permutation of $\mathbf{a}$. We will use this permutation to construct our codes.
\label{remark2}
\end{remark}

In this code, we have $\LL = \MM/k = (n-k)^{k}$, so that the $k$ sources, $\mathbf{a}_{1}, \mathbf{a}_{2},\ldots, \mathbf{a}_{k}$ are all $(n-k)^{k} \times 1$ vectors and the coding sub-matrices are $(n-k)^k \times (n-k)^{k}$ matrices. 

Consider the permutation matrix $\mathbf{P}_{i}$ defined as 
\begin{equation} \mathbf{P}_{i} = \left(\begin{array}{c} \mathbf{e}\big(\langle {\lbeta}_{i}(1)\rangle\big) \\ \mathbf{e}\big(\langle{\lbeta}_{i}(2)\rangle\big) \\ \vdots \\\mathbf{e}\big(\langle{\lbeta}_{i}((n-k)^{k})\rangle\big)  \end{array}\right)\label{eq:Permutation}\end{equation}
for $i=1,2, \ldots, k$, where $\mathbf{e}(1), \mathbf{e}(2),\ldots, \mathbf{e}((n-k)^{k})$ are the rows of the identity matrix $\mathbf{I}_{(n-k)^{k}}$. Note that because of Remark \ref{remark2}, the above matrix is indeed a permutation matrix. Then, the coding sub-matrices are defined as 
$$ \mathbf{C}_{j,i} = \lambda_{j,i} \mathbf{P}_i^{j-k-1}.$$
Thus, to understand the structure of the above permutation, consider an arbitrary column vector $$\mathbf{a} = \left(a(1)~~a(2)~~\ldots a((n-k)^{k}) \right)^{T}.$$ Then, let $j = \langle x_1,x_2,x_3,\ldots,x_{k})\rangle$ for $1 \leq j \leq (n-k)^{k}$. Then, the $j$th component of $\mathbf{P}_{i} \mathbf{a}$ is $$a(\langle (x_1, x_2, \ldots, x_{i-1}, x_{i}\oplus 1, x_{i+1},\ldots,x_{k})\rangle).$$ Thus, we can write
\begin{eqnarray*} d_{k+r+1}(\langle x_1,x_2,\ldots,x_k\rangle) &=& \lambda_{k+r+1,1} a_1(x_1\oplus r, x_2,x_3,\ldots,x_k) + \lambda_{k+r+1,2} a_2(x_1, x_2\oplus r,x_3,\ldots,x_k) \\
\\ &&+ \ldots + \lambda_{k+r+1,k} a_k(x_1, x_2,x_3,\ldots,x_k\oplus r) 
\end{eqnarray*}
where $r \in \{0,1,2,\ldots, n-k-1\}$. This describes the coding sub-matrices.

Now, in case of failure of node $l$, the rows of the repair matrices $\mathbf{V}_{l}$ are chosen from the set $\{ \mathbf{e}(m): \gbeta_{l}(m) = 0\}$. Since $\gbeta_{l}(m)$ can take $n-k$ values, this construction has $\frac{\LL}{k} = (n-k)^{k-1}$ rows for $\mathbf{V}_{l}$ as required. Because of the construction, we have the following interference alignment relation for $i \neq l, j \in \{k+1,k+2,\ldots,n\}$
$$ \mbox{rowspan}(\mathbf{C}_{j,i} \mathbf{V}_{l}) = \mbox{rowspan}(\{\mathbf{e}(m):\gbeta_{l}(m)=0\}).$$
Further, 
$$ \mbox{rowspan}(\mathbf{C}_{j,l} \mathbf{V}_{l}) = \mbox{rowspan}(\{\mathbf{e}(m):\gbeta_{l}(m)=j-k-1\}).$$
for $j \in \{k+1,k+2,\ldots,n\}$ so that (\ref{eq:reconstruction}) is satisfied and the desired signal can be reconstructed from the interference.
All that remains to be shown is the MDS property. This is shown in Appendix \ref{app:MDS}.

\section{Explicit Construction of Codes for $n-k \in \{2,3\}$}
\label{sec:explicit}
While, theoretically, any $(n,k)$ MDS code could be used to build distributed storage systems, in practice, the case of having a small number of parity nodes, i.e. small values of $n-k$, is especially of interest. In fact, a significant portion of literature on use of codes for storage systems is devoted to building codes for the cases of $(n-k)\in \{2,3\}$ with desirable properties (See, for example, \cite{EVENODD, LIBERATION, RDP, STAR}). While these references focused on constructing MDS codes with efficient encoding and decoding properties, here, we study the construction of MDS codes for $n-k \in \{2,3\}$ with desirable repair properties. 

In the previous section, we provided \emph{random} code constructions based on permutation matrices. In this section, we further strengthen our constructions by providing \emph{explicit} code constructions for the important case of $n-k \in \{2,3\}$. Note that the codes constructed earlier were random constructions because of the fact that scalars $\lambda_{j,i}$ were picked randomly from the field. Further, note that, as long as $\lambda_{j,i}, j=k+1,k+2,\ldots,n, i=1,2,\ldots,k$ are \emph{any} set of non-zero scalars, the repair bandwidth for failure of a single systematic node is $\frac{n-1}{n-k}$ units as required. The randomness of the scalars $\lambda_{j,i}$ was used in the previous section to show the existence of codes which satisfy the MDS property. In this section, for the two cases of $n-k=2$ and $n-k=3$, we choose these scalars explicitly (i.e., not randomly) so that the MDS property is satisfied. For both cases, the scalars $\lambda_{i,j}$ are chosen as 
\begin{eqnarray}
\label{eq:2parity}
\lambda_{j,i} = \lambda_{j}^{i-1}
\end{eqnarray}
so that we have 
$$ \mathbf{C}_{j,i} = (\lambda_{j} \mathbf{P}_{j})^{i-1}$$
for $j=\{1,2,\ldots,n-k\}, i=1,2\ldots,k$.  

If $n-k=2$, we choose $q\geq (2k+1)$ and choose non-zero scalars $\lambda_{1},\lambda_{2},\ldots,\lambda_{k}$ from the field so that 
$$\lambda_{i}\neq \lambda_{j}, \lambda_{i} + \lambda_{j} \neq 0, \mbox{for }i \neq j.$$ Note that in a field of size $(2k+1)$ or bigger, scalars satisfying the above can be chosen by ensuring that $\lambda_{i}^{'} + \lambda_{i} = 0 \Rightarrow \lambda_{i}^{'} \notin \{\lambda_{1}, \lambda_{2},\ldots,\lambda_{k}\}$. . With this choice of scalars, in Appendix \ref{app:explicit}, we show that the code satisfies the MDS property. 

For $n-k=3$, we choose $\lambda_{1},\lambda_{2}, \ldots, \lambda_{k}$ to be $k$ non-zero elements in the field $\mathbb{F}_{q},$ where $q \geq 2k+1$ is a prime, so that 
\begin{equation} \lambda_{i} \neq \lambda_{j}, \lambda_{i}+\lambda_{j} \neq 0 
\label{eq:3parity}\end{equation}
for all $i \neq j, i,j \in \{1,2,\ldots,k\}$. Note that elements $\lambda_{i}$ satisfying the above conditions can be chosen satisfying the above properties if $q \geq 2k+1.$ In Appendix \ref{app:explicit}, we also show that the code described here for the case of $n-k =3$ is an MDS code.

\section{A Subspace Interference Alignment Framework for Optimal Repair}
\label{sec:Tensors}
In this section, we return to Problem 2 described in Section \ref{sec:problem}. Before we consider this problem, we summarize some properties of tensor (Kronecker) products below; the notation $\otimes$ is used to denote the tensor (Kronecker) product between two matrices.
\begin{itemize}
\item Mixed Product Property: $$(\mathbf{P}_{1} \otimes \mathbf{P}_{2} \ldots \otimes \mathbf{P}_{m}) (\mathbf{Q}_{1} \otimes \mathbf{Q}_{2} \ldots \otimes \mathbf{Q}_{m}) = (\mathbf{P}_{1} \mathbf{Q}_{1}) \otimes (\mathbf{P}_{2}\mathbf{Q}_{2})  \ldots \otimes (\mathbf{P}_{m}\mathbf{Q}_{m}) $$
\item  Invariance w.r.t span:\\
\emph{If all the factors of a tensor product align with the corresponding factors of another tensor product, then the corresponding products also align, and vice-versa.}
Formally, let $\mathbf{P}_{i},\mathbf{Q}_{i},i=1,2,\ldots,m$ be matrices such that the dimension of $\mathbf{P}_{i}$ is equal to the dimension of $\mathbf{Q}_{i}$. Then,  $\mbox{rowspan}(\mathbf{P}_{i}) = \mbox{rowspan}(\mathbf{Q}_{i}) \neq \{\mathbf{0}\} , i=1,2,\ldots,m$, if and only if
$$\mbox{rowspan}(\mathbf{P}_{1} \otimes \mathbf{P}_{2} \ldots \otimes \mathbf{P}_{m}) = \mbox{rowspan}(\mathbf{Q}_{1} \otimes \mathbf{Q}_{2} \ldots \otimes \mathbf{Q}_{m}),$$
where $\mathbf{0}$ represents the row vector whose entries are all equal to $0$. \item Inheritance of linear independence: \\
\emph{If the rows of one of the factors of a tensor product is linearly independent of the rows of the corresponding factor in another tensor product, then the rows of the corresponding products are also linearly independent.} More formally, let $\mathbf{P}_{i},\mathbf{Q}_{i},i=1,2,\ldots,m$ be matrices such that the dimension of $\mathbf{P}_{i}$ is equal to the dimension of $\mathbf{Q}_{i}$. Now, suppose that $\mbox{rowspan}(\mathbf{P}_{l}) \cap \mbox{rowspan}(\mathbf{Q}_{l})=\{\mathbf{0}\}$ for some $l \in \{1,2,\ldots,m\},$ i.e., each row of $\mathbf{P}_{l}$ is linearly independent of all the rows of $\mathbf{Q}_{l}$ for some $l\in\{1,2,\ldots,m\}$. Then, 
$$\mbox{rowspan}(\mathbf{P}_{1} \otimes \mathbf{P}_{2} \otimes \ldots \otimes \mathbf{P}_{m} )\cap \mbox{rowspan}(\mathbf{Q}_{1} \otimes \mathbf{Q}_{2} \otimes \ldots \otimes \mathbf{Q}_{m}) = \{\mathbf{0}\}$$
\end{itemize}
The second and third properties above follow as a result of bilinearity and associativity of tensor products. 
The above properties were used in \cite{Suh_Tse_subspace} to develop a type of interference alignment called \emph{subspace interference alignment} in the context of cellular networks. In subspace interference alignment, the property of the invariance of tensor products w.r.t. span plays a central role in ensuring that interference aligns, and the inheritance of linear independence property plays a central role in ensuring that desired signals are linearly independent of the interference. This intuition recurs in our application of the concept here. We apply this idea of subspace interference alignment in the context of the repair problem - specifically, we use the idea of subspace interference alignment in the context of Problem $2$ in Section \ref{sec:problem}. We use $N=3$ here to demonstrate the main idea - the framework developed here can be used to solve the problem for any $N \in \mathbb{N}$. For the convenience of the reader, equations (\ref{eq:1a})-(\ref{eq:3a}) associated with the problem are restated (albeit in a slightly different, but equivalent, form) here. 
\begin{eqnarray}
\mbox{rowspan} (\mathbf{V}_{i} \mathbf{H}_{j}) &=& \mbox{rowspan}({\mathbf{V}_i}), j\in\{1,2,3\}-\{i\}\label{eq:1b}\\
 \mbox{rowspan}(\mathbf{V}_{i}) \cap \mbox{rowspan}(\mathbf{V}_{i}\mathbf{H}_{i}) &=& \{\mathbf{0}\}\label{eq:2b}\\
\mbox{rank}(\mathbf{V}_{i}) = \mbox{rank}(\mathbf{H}_{i})/2 &=& L/2\label{eq:3b}
\end{eqnarray}
where $\mathbf{0}$ is the $1 \times L$ row vector of zeros. To recollect, as shown in Section \ref{sec:connections}, a solution to the above problem can lead to an $n=N+2,k=N$ code by choosing coding sub-matrices $\mathbf{C}_{k+1,i}=\mathbf{I}_{L}$ and $\mathbf{C}_{k+2,i} = \mathbf{H}_{i}$ for $i=1,2,\ldots,k$.

\subsection{Simplifying the above problem}

In the remainder of this section, we will use the properties of the tensor products listed to simplify the above problem to the following: Find $\mathbf{U}_{0},\mathbf{G}_{0},\mathbf{G}_{1}$ such that
\begin{eqnarray}
\mbox{rowspan} (\mathbf{U}_{0} \mathbf{G}_{0}) &=& \mbox{rowspan}({\mathbf{U}_0}), \label{eq:1simple}\\
 \mbox{rowspan}(\mathbf{U}_{0}) \cap \mbox{rowspan}(\mathbf{U}_{0}\mathbf{G}_{1}) &=& \{\mathbf{0}\}\label{eq:2simple}
\end{eqnarray}
where $\mathbf{U}_{0}$ is a $1 \times 2$ row vector, $\mathbf{G}_{0},\mathbf{G}_{1}$ are $2 \times 2$ matrices and $\mathbf{0}$ is a $1 \times 2$ vector of zeros.  In other words, the problem finding $N$ matrices $\mathbf{V}_{i},\mathbf{H}_{i},i=1,2,\ldots,N$ satisfying all the relations represented in (\ref{eq:1b}),(\ref{eq:2b}) can be simplified into finding three matrices $\mathbf{U}_{0},\mathbf{G}_{0},\mathbf{G}_{1}$ satisfying (\ref{eq:1simple}),(\ref{eq:2simple}). Note that finding $\mathbf{U}_{0},\mathbf{G}_{0},\mathbf{G}_{1}$ satisfying the above is straightforward, and the eigen-vector approach used in Problem $1$ in Section \ref{sec:problem} works, i.e., we can pick the matrices so that $\mathbf{G}_{0}^{T}$ and $\mathbf{G}_{1}^{T}$ do not have a common eigen vector and pick $\mathbf{U}_{0}^{T}$ to be an eigen vector of $\mathbf{G}_{0}^{T}$.  As we show next, we use tensor products to ``stitch together'' $N$ independent instances of the simpler problem of satisfying (\ref{eq:1simple}),(\ref{eq:2simple}), to find matrices satisfying (\ref{eq:1b})-(\ref{eq:3b}).

In our solution to (\ref{eq:1b})-(\ref{eq:3b}), we have $L=2^{N}=8$. Suppose we restrict the matrices (\ref{eq:1b})-(\ref{eq:3b}) to have the following structure.
$$\mathbf{H}_{i} = \lambda_{i} (\mathbf{G}_{i,1} \otimes \mathbf{G}_{i,2} \otimes \mathbf{G}_{i,3})$$
where $\mathbf{G}_{i,j}$ is a $2 \times 2$ full rank matrix for $j=1,2,3, i=1,2,3$ and $\lambda_{i}$ is some non-zero scalar over the field $\mathbb{F}_{q}$. Note that $\mathbf{H}_{i}$ has a dimension of $8 \times 8$ and a rank of $8$ as required, with the full rank property coming from the fact that $\mathbf{G}_{i,j},j=1,2,3$ each has a rank of $2$. 
We also choose
$$\mathbf{V}_{i} = \mathbf{U}_{i,1} \otimes \mathbf{U}_{i,2} \otimes \mathbf{U}_{i,3}$$ 
where $\mathbf{U}_{1,1},\mathbf{U}_{2,2},\mathbf{U}_{3,3}$ are $1 \times 2$ row vectors. $\mathbf{U}_{i,j}$ for $i \neq j$ are $2 \times 2$ matrices having a full rank of $2$. Note that, with this choice of dimensions, $\mathbf{V}_{i}$ have a dimension of $4 \times 8$ as required. Now, we intend to choose matrices $\mathbf{G}_{i,j},\mathbf{U}_{i,j}, i,j \in\{1,2,3\}$ to satisfy (\ref{eq:1b}) and (\ref{eq:2b}). We choose these matrices to satisfy 
\begin{eqnarray} \mbox{rowspan}(\mathbf{U}_{i,i} \mathbf{G}_{j,i}) = \mbox{rowspan}(\mathbf{U}_{i,i}), i \neq j \label{eq:a}\\
\mbox{rowspan}(\mathbf{U}_{i,i} \mathbf{G}_{i,i}) \cap \mbox{rowspan}(\mathbf{U}_{i,i})=\{\mathbf{0}\}\label{eq:b}\end{eqnarray}
for $i,j=1,2,3$. In other words, the $1 \times 2$ row vector $\mathbf{U}_{i,i}$ is invariant w.r.t $\mathbf{G}_{j,i},j\neq i$ but is linearly independent of $\mathbf{U}_{i,i}\mathbf{G}_{i,i}$. The $2 \times 2$ matrices $\mathbf{U}_{i,j},i \neq j$ can be chosen to be arbitrary full rank matrices. Equation (\ref{eq:a}) ensures that (\ref{eq:1b}) is satisfied by using the invariance of tensor prodcts w.r.t. span, i.e., by ensuring that each of the $N=3$ factors on the left hand side of (\ref{eq:1b}) align with the space spanned by the corresponding factor on the right hand side. To see this, note the following.
\begin{eqnarray*}
\mbox{rowspan}(\mathbf{V}_{1} \mathbf{H}_{2}) &=& \mbox{rowspan}\left((\mathbf{U}_{1,1} \otimes \mathbf{U}_{1,2}\otimes \mathbf{U}_{1,3})(\mathbf{G}_{2,1}\otimes \mathbf{G}_{2,2} \otimes \mathbf{G}_{2,3})  \right)
\\&\stackrel{(a)}{=}& \mbox{rowspan}(\mathbf{U}_{1,1} \mathbf{G}_{2,1} \otimes \mathbf{U}_{1,2} \mathbf{G}_{2,2} \otimes \mathbf{U}_{1,3} \mathbf{G}_{2,3}) \\
&\stackrel{(b)}{=}& \mbox{rowspan}(\mathbf{U}_{1,1}\mathbf{G}_{2,1} \otimes \mathbf{U}_{1,2} \otimes \mathbf{U}_{1,3} \mathbf{G}_{2,3}) \\
&\stackrel{(c)}=& \mbox{rowspan}(\mathbf{U}_{1,1} \otimes \mathbf{U}_{1,2} \otimes \mathbf{U}_{1,3}) \\
&=& \mbox{rowspan}(\mathbf{V}_{1})
\end{eqnarray*}
$(a)$ follows from the Mixed Product Property of tensor products. $(b)$ follows from (\ref{eq:a}) and the invariance of the tensor product w.r.t. span. Similarly $(c)$ follows from invariance of tensor products w.r.t span, and the fact that $\mbox{rowspan}(\mathbf{U}_{i,j})=\mbox{rowspan}(\mathbf{U}_{i,j}\mathbf{G}_{m,n})$ for $i \neq j,$ which in turn, follows from the fact that $\mathbf{U}_{i,j}$ and $\mathbf{G}_{m,n}$ are full rank matrices for $i \neq j$.  Thus, (\ref{eq:a}) ensures that (\ref{eq:1b}) is satisfied for all $i=1,2,\ldots,N$. 

Similarly, we show below that as long as (\ref{eq:b}) holds, equation (\ref{eq:2b}) is satisfied because of the inheritance of linear independence property of tensor products.
\begin{eqnarray*}
\mbox{rowspan}(\mathbf{V}_{1} \mathbf{H}_{1}) &=& \mbox{rowspan}\left((\mathbf{U}_{1,1} \otimes \mathbf{U}_{1,2}\otimes \mathbf{U}_{1,3})(\mathbf{G}_{1,1}\otimes \mathbf{G}_{1,2} \otimes \mathbf{G}_{1,3})  \right)\\
&{=}& \mbox{rowspan}(\mathbf{U}_{1,1} \mathbf{G}_{1,1} \otimes \mathbf{U}_{1,2} \mathbf{G}_{1,2} \otimes \mathbf{U}_{1,3} \mathbf{G}_{1,3}) 
\\&{=}& \mbox{rowspan}(\mathbf{U}_{1,1} \mathbf{G}_{1,1} \otimes \mathbf{U}_{1,2} \otimes \mathbf{U}_{1,3}) \\
\Rightarrow  \mbox{rowspan}(\mathbf{V}_{1} \mathbf{H}_{1}) \cap \mbox{rowspan}(\mathbf{V}_{1}) &=&\mbox{rowspan}(\mathbf{U}_{1,1} \mathbf{G}_{1,1} \otimes \mathbf{U}_{1,2} \otimes \mathbf{U}_{1,3}) \cap \mbox{rowspan}(\mathbf{U}_{1,1}  \otimes \mathbf{U}_{1,2} \otimes \mathbf{U}_{1,3} )\\
&=& \{\mathbf{0}\}
\end{eqnarray*}
where the final equation follows from (\ref{eq:b}) and the inheritance of linear independence into tensor products. Now, we have reduced the task of finding matrices satisfying (\ref{eq:1b})-(\ref{eq:3b}) to finding matrices satisfying (\ref{eq:a})-(\ref{eq:b}). Suppose we set 
$$\mathbf{U}_{1,1}=\mathbf{U}_{2,2}=\mathbf{U}_{3,3}=\mathbf{U}_{0}$$
$$\mathbf{G}_{1,1}=\mathbf{G}_{2,2}=\mathbf{G}_{3,3}=\mathbf{G}_{1}$$
$$\mathbf{U}_{i,j}=\mathbf{U}_{1}, \mathbf{G}_{i,j}=\mathbf{G}_{0},i\neq j$$
Now equations (\ref{eq:a}),(\ref{eq:b}) essentially boil down to finding matrices
\begin{equation}\mbox{rowspan}(\mathbf{U}_{0} \mathbf{G}_{0})=\mbox{rowspan}(\mathbf{U}_{0})\label{eq:a1}\end{equation}
\begin{equation}\mbox{rowspan}(\mathbf{U}_{0} \mathbf{G}_{1})\cap \mbox{rowspan}(\mathbf{U}_{0}) = \{\mathbf{0}\}\label{eq:b1}\end{equation}
$\mathbf{U}_{1}$ can be any full rank $2 \times 2$ matrix. Thus, the simplification of Problem $2$ of Section \ref{sec:problem} is complete (at least, for $N=3$). As discussed before, the eigen vector approach illustrated for Problem $1$ in Section \ref{sec:problem} suffices to finding $\mathbf{U}_{0},\mathbf{G}_{0},\mathbf{G}_{1}$ satisfying the above relations. In fact, to obtain the $(5,3)$ permutations-based coding sub-matrices described previously, we choose 
$$\mathbf{U}_{0} = \left(1~~0\right)$$
$$\mathbf{G}_{1} = \left(\begin{array}{cc}0 & 1 \\ 1 & 0 \end{array} \right)$$
$$\mathbf{U}_{1}=\mathbf{G}_{0} = \mathbf{I}_{2}, i,j \in \{1,2,3\}, i \neq j$$
It can be noticed that the matrices $\mathbf{U}_{0},\mathbf{G}_{1},\mathbf{U}_{1}$ satisfy (\ref{eq:a})-(\ref{eq:b}). Further, in general, any choice of matrices which satisfy (\ref{eq:1b})-(\ref{eq:3b}), and hence (\ref{eq:a1}),(\ref{eq:b1}) would solve problem $2$, and hence, can be used for codes with optimal repair bandwidth for distributed storage, for $n-k=2$. For example, we could alternately  the matrices inspired by ergodic alignment \cite{Nazer_Gastpar_Jafar_Vishwanath}. These matrices are shown below.
$$\mathbf{U}_{0} = \left(1~~-1\right)$$
$$\mathbf{G}_{1} = \left(\begin{array}{cc}1&0 \\ 0 & -1\end{array}\right)$$
$$\mathbf{G}_{0} = \mathbf{I}_{2}, i,j \in \{1,2,3\}, i \neq j$$
and $\mathbf{U}_{1}$ to be any arbitrary full rank $2\times 2$ matrix. In fact, this choice of matrices has been studied for efficiently repairable code constructions in \cite{Dimitris_Dimakis_Hadamard}.
\subsection{Discussion}
\begin{itemize}
\item For $N=3$, we used $L=2^{3}$ and expressed $\mathbf{H}_{i}$ as a Kronecker product of $N=3$ matrices. For an arbitrary $N,$ we can extend the above framework by expressing $\mathbf{H}_{i}$ as a Kronecker product of $N$ $2 \times 2$ matrices so that $L=2^{N}$. $\mathbf{V}_{i}$ is also, similarly, a Kronecker product of $N$ matrices, such that the $i$th matrix is a $1 \times 2$ matrix, and the remaining $N-1$ matrices participating in the Kronecker product, are $2 \times 2$ matrices. 

\item Because of Section \ref{sec:connections}, the subspace interference alignment framework here can be used to generate $(k+2,k)$ codes which can be repaired by downloading $1/2$ the data stored in every surviving node. This is because equation (\ref{eq:1b}) ensures that the interference is aligned, and (\ref{eq:2b}) ensures that the lost (desired) symbols can be reconstructed from the downloaded data. However, this framework does not ensure Property \ref{property:MDS}, i.e., it does not ensure that the code generated is MDS. The MDS property can be ensured by choosing the scalars $\lambda_{i}$ randomly over the field and using the Schwartz-Zippel Lemma along the same lines as the proof in Appendix \ref{app:MDS}. In other words, the Schwartz-Zippel Lemma ensures that there exist at least one choice of scalars $\lambda_{i}$ so that the code is an MDS code.

\item The problems motivated in Section \ref{sec:problem} and solved in this section are related to optimal repair of failed systematic nodes in a distributed storage system with $n-k=2$ parity nodes. In general, if $n-k>2,$ the framework developed here can be used to show that the problem of finding repair-bandwidth optimal MDS codes can be decomposed into the problem of finding full rank $(n-k)\times(n-k)$ matrices $\mathbf{G}_{0},\mathbf{G}_{1},\ldots,\mathbf{G}_{n-k-1}$ and $1 \times (n-k)$ dimensional row vector $\mathbf{U}_{0}$ such that 
\begin{eqnarray} \mbox{rowspan}(\mathbf{U}_{0} \mathbf{G}_{0}) &=& \mbox{rowspan}(\mathbf{U}_{0})\label{eq:gen_simplified1}\\
\mbox{rank}\left(\left[\begin{array}{c}\mathbf{U}_{0} \\ \mathbf{U}_{0} \mathbf{G}_{1}\\ \vdots\\\mathbf{U}_{0}\mathbf{G}_{n-k-1}\end{array} \right]\right) &=& n-k\label{eq:gen_simplified2}
 \end{eqnarray}
With a solution to the above problem, the coding sub-matrices can be chosen as $\mathbf{C}_{k+1,j}=\mathbf{I}$ and for $m > 1,$
$$\mathbf{C}_{k+m,j} = \underbrace{\mathbf{G}_{0} \otimes \ldots \otimes \mathbf{G}_{0}}_{(j-1)\mbox{ times}} \otimes  \mathbf{G}_{m-1} \otimes \underbrace{\mathbf{G}_{0} \otimes \ldots \otimes \mathbf{G}_{0}}_{k-j \mbox{ times}}.$$
The repair matrices $\mathbf{V}_{j}$ can be obtained as 
$$\mathbf{V}_{j} = \underbrace{\mathbf{U}_{1} \otimes \ldots \otimes \mathbf{U}_{1}}_{j-1\mbox{ times}} \otimes  \mathbf{U}_{0} \otimes \underbrace{\mathbf{U}_{1} \otimes \ldots \otimes \mathbf{U}_{1}}_{k-j\mbox{ times}}$$
where $\mathbf{U}_{1}$ is a full rank $(n-k) \times (n-k)$ matrix. The permutation matrices used for code development in Section \ref{sec:Permutation} here can be interpreted as one solution to equations (\ref{eq:gen_simplified1})-(\ref{eq:gen_simplified2}).
\item It is worth noting that the framework developed in this section can be used to generate that codes are optimal from the perspective of the repair bandwidth. However, the codes developed need not be optimal from the perspective of the amount of disk access in the storage system. The codes of Section \ref{sec:Permutation} which fit within this framework, satisfy the additional property of being optimal from the perspective of the amount of disk access.
\end{itemize}

\section{Conclusion}
In this paper, we construct class of MDS codes based on subspace interference alignment with optimal repair bandwidth for a single failed systematic node. A class of our code constructions are optimal, not only in terms of repair bandwidth, but also in terms of the amount of disk access during the recovery of a single failed node. Since we effectively provide the first set of repair-bandwidth optimal MDS codes for arbitrary $(n,k)$, this work can be viewed as a stepping stone towards implementation of MDS codes in distributed storage systems.


From the perspective of storage systems, there remain several unanswered questions. First, there remains open the existence of finite codes which can achieve more efficient repair of parity nodes as well, along with systematic nodes. Second, we assume that the new node connects to \emph{all} $d=n-1$ surviving nodes in the system. An interesting question is whether finite code constructions can be found to conduct efficient repair when the new node is restricted to connect to a subset of the surviving nodes. While asymptotic constructions satisfying the lower bounds have been found for both these problems, the existence of finite codes satisfying these properties remain open. Finally, the search for repair strategies of existing codes, which is analogous to the search of interference alignment beamforming vectors for fixed channel matrices in the context of interference channels, remains open. While iterative techniques exist for the wireless context \cite{Gomadam_Cadambe_Jafar, Peters_Heath}, they cannot be directly extended to the storage context because of the discrete nature of the optimization problem in the latter context. Such algorithms, while explored in the context of certain classes of codes in \cite{Dimakis_EVENODD, RDP_repair}, remain an interesting area of future work.



\appendices
\section{MDS Property}
\label{app:MDS}
We intend to show that the determinant of the matrix in (\ref{eq:MDS}) is a non-zero polynomial in $\Lambda=\{\lambda_{j,i}, j=k+1,k+2, \ldots, n, i=1,2,\ldots,k\}$ for any $j_1,j_2,\ldots, j_k \in \{1,2,\ldots,n\}$. If we show this, then, each MDS constraint corresponds to showing that a polynomial $p_{j_1,j_2,\ldots,j_k}(\Lambda)$ is non-zero. Using the Schwartz-Zippel Lemma on the product of these polynomials $\pi_{j_1,j_2,\ldots,j_k} p_{j_1,j_2,\ldots,j_k}(\Lambda)$ automatically implies the existence of $\Lambda$ so that the MDS constraints are satisfied, in a sufficiently large field. Therefore, all that remains to be shown is that the determinant of (\ref{eq:MDS}) is a non-zero polynomial in $\Lambda$. We will show this by showing that there exists at least one set of values for the variables $\Lambda$ such that the determinant of (\ref{eq:MDS}) is non-zero. To show this, we first assume, without loss of generality, that $j_1, j_2, \ldots, j_k$ are in ascending order. Also, let $j_1, j_2, \ldots, j_{k-m} \in \{1,2,\ldots, k\}$ and $j_{k-m+1}, j_{k-m+2}, \ldots, j_{k} \in \{k+1,k+2,\ldots,n\}$. For simplicity we will assume that $j_1 =1, j_2=2, \ldots, j_{k-m} = k-m$. The proof for any other set $\{ j_1, j_2,\ldots,j_{k-m}\}$ is almost identical to this case, except for a difference in the indices used henceforth. Substituting the appropriate values of $\mathbf{C}_{j,i}$, the matrix in (\ref{eq:MDS}) can be written as 

\begin{equation}
\left[ \begin{array}{ccccc}
\mathbf I& \ldots & \mathbf 0& \ldots &\mathbf 0\\
\mathbf 0& \ldots & \mathbf 0& \ldots &\mathbf 0\\
\vdots & \ddots & \vdots & \ddots & \vdots \\
\mathbf 0& \ldots & \mathbf I& \ldots &\mathbf 0\\
\mathbf\lambda_{j_{k-m+1},1}\mathbf{P}_{1}^{s_{k-m+1}}& \ldots &\lambda_{j_{k-m+1},k-m}\mathbf{P}_{k-m}^{s_{k-m+1}}& \ldots & \lambda_{j_{k-m+1},k}\mathbf{P}_{k}^{s_{k-m+1}}\\
\vdots& \vdots & \vdots & \ddots & \vdots\\
\lambda_{j_{k},1}\mathbf{P}_{1}^{s_k} & \ldots & \lambda_{j_k,k-m}\mathbf{P}_{k-m}^{s_{k}}& \ldots & \lambda_{j_k,k-m}\mathbf{P}_{k}^{s_k} \end{array} \right] \label{eq:MDSproof}
\end{equation}
where $s_{i}=j_{i}-k-1$. 
Now, if $$\lambda_{j,i} = \left\{\begin{array}{cc}0 & \mbox{if }(j,i) \notin \{(j_t,t) : t=k-m+1, k-m+2,\ldots,k \} \\ 1 & \mbox{otherwise} \end{array}\right\}$$ then the above matrix is a block diagonal matrix. Therefore, its determinant evaluates to the product of the determinant of its diagonal blocks, i.e., $\displaystyle\prod_{u=k-m+1}^{k}|\mathbf{P}_{u}^{s_u}|$ which is non-zero. This implies that the determinant in (\ref{eq:MDS}) is a non-zero polynomial in $\Lambda$ as required. This completes the proof. 

\section{Proof of MDS property for explicit constructions of Section \ref{sec:explicit}}
\label{app:explicit}
We need to show Property \ref{property:MDS}. Before we show this property, we begin with the following Lemma which shows that the coding submatrices in our constructions commute. 
\begin{lemma}
\label{lemma:commute}
$$\mathbf{P}_{i}^{m_1}\mathbf{P}_{j}^{m_2} = \mathbf{P}_{j}^{m_2}\mathbf{P}_{i}^{m_1}$$
where $\mathbf{P}_{i}$ is chosen as in (\ref{eq:Permutation}).
\end{lemma}
\begin{proof}
In order to show this, we show that $\mathbf{P}_{i}\mathbf{P}_{j}\mathbf{a} = \mathbf{P}_{j} \mathbf{P}_{i} \mathbf{a}$ for any $2^{k} \times 1$ dimensional column vector $\mathbf{a}$. Assuming without loss of generality that $i < j$, this can be seen by verifying that 
$$\mathbf{P}_{i}\mathbf{P}_{j} \mathbf{a} = \mathbf{P}_{j} \mathbf{P}_{i} \mathbf{a} = \left[\begin{array}{c} a(\langle \overbrace{0,0,\ldots,0}^{\mbox{$i-1$ entries}}, 0\oplus m_1, \overbrace{0,\ldots,0}^{j-i-1\mbox{ entries}},0 \oplus m_2,0,\ldots,0\rangle)\\ a(\langle 1,1,\ldots,1, 1\oplus m_1, 1,\ldots,1,1\oplus m_2,1,\ldots,1\rangle)\\ \vdots \\ a(\langle k-1,k-1,\ldots,k-1, (k-1)\oplus m_1, 0,\ldots,0,(k-1)\oplus m_2,k-1,\ldots,k-1\rangle)   \end{array}{c}\right]$$
In other words, the $<r_1,r_2,\ldots,r_k>$th element of both $\mathbf{P}_{i}^{m_1} \mathbf{P}_{j}^{m_2} \mathbf{a}$ and $\mathbf{P}_{j}^{m_2} \mathbf{P}_{i}^{m_1} \mathbf{a}$ can be verified to be $$a(<r_1,r_2,\ldots,r_{i-1},r_{i}\oplus m_1, r_{i+1},\ldots,r_{j-1},r_{j}\oplus m_2,r_{j+1},\ldots,r_k\rangle)$$
\end{proof}
Now, we proceed to show the \ref{property:MDS} property for $n-k \in \{2,3\}$. Without loss of generality, we assume that that $j_1, j_2, \ldots, j_k$ are in ascending order.
\subsubsection*{Case 1: $n-k=2$}
We divide this case into $2$ scenarios. In the first scenario , $j_1, j_2,\ldots, j_{k-1} \in \{1,2,\ldots,k\}$ and $j_{k} \in \{k+1, k+2\}$. Note that this corresponds to reconstructing the data from $k-1$ systematic nodes and a single parity node. Now, substituting this in equation (\ref{eq:MDSproof}) in Appendix \ref{app:MDS}, and expanding this determinant along the first $(k-1)\LL$ columns, we get this determinant to be equal to $|\mathbf{C}_{j_k,i}|$. Therefore, the desired property is equivalent to the matrix $\mathbf{C}_{j,i}=(\lambda_{i} \mathbf{P}_{i})^{j-k-1}$ to be full rank for all $j \in \{k+1,k+2,\ldots,n\}, i=1,2,\ldots,k$. This scenario is hence, trivial. 
Now, in the second scenario, consider the case where $j_1, j_2,\ldots, j_{k-2} \in \{1,2,\ldots,k\}$ and $j_{k-1} =k+1, j_{k}= k+2$. This corresponds to the case where the original sources are reconstructed using $k-2$ systematic nodes, and both parity nodes. By substituting in (\ref{eq:MDSproof}) and expanding along the first $(k-2) \LL$ rows, the MDS property can be shown to be equivalent showing that the matrix
$$ \left[\begin{array}{cc} 
\mathbf{I} & \mathbf{I}\\
\lambda_{i}\mathbf{P}_{i} &  \lambda_{j} \mathbf{P}_{j}
\end{array} \right]$$
having full rank. Now, note that the matrices $\mathbf{P}_{i}$ and $P_{j}$. On noting that the determinant of commuting block-matrices can be evaluated by using the element-wise determinant expansion over blocks \cite{Commuting_Matrices}, the determinant of the matrix can be written as $$|\lambda_{j}\mathbf{P}_{j}-\lambda_{i} \mathbf{P}_{i}| =  \lambda_{j}^{-1} |\mathbf{P}_{i}^{-1}| |\mathbf{P}_{j}\mathbf{P}_{i}^{-1}-\lambda_{i}\lambda_{j}^{-1}\mathbf{I}|.$$
Note that the above expression is equal to $0$ if and only if $\lambda_{i}\lambda_{j}^{-1}$ is an Eigen-value of the permutation matrix $\mathbf{P}_{j}\mathbf{P}_{i}^{-1}$. Note here that $\mathbf{P}_{j}\mathbf{P}_{i}^{-1}$ is a permutation matrix whose square is the identity matrix. Therefore, the only possible eigen values of this matrix are the square roots of unity, i.e., $1$ and $-1$. As noted in (\ref{eq:2parity}), we have $\lambda_{i} \neq \lambda_{j}, \lambda_{i}+\lambda_{j} \neq 0 \Rightarrow \lambda_{i}\lambda_{j}^{-1} \neq 1, \lambda_{i}\lambda_{j}^{-1} \neq -1$, and hence, the determinant shown above is non-zero and the matrix is full-rank as required.
\subsubsection*{Case 2: $n-k=2$}
We divide this case into $5$ scenarios as listed below.
\begin{enumerate}
\item $j_1, j_2,\ldots, j_{k-1} \in \{1,2,\ldots,k\}$ and $j_{k} \in \{k+1, k+2,k+3\}$.
\item $j_1, j_2,\ldots, j_{k-2} \in \{1,2,\ldots,k\}$ and $j_{k-1}= k+1, j_{k}=k+2\}$.
\item $j_1, j_2,\ldots, j_{k-2} \in \{1,2,\ldots,k\}$ and $j_{k-1}= k+2, j_{k}=k+3\}$.
\item $j_1, j_2,\ldots, j_{k-2} \in \{1,2,\ldots,k\}$ and $j_{k-1}= k+1, j_{k}=k+3\}$.
\item $j_1, j_2,\ldots, j_{k-3} \in \{1,2,\ldots,k\}$ and $j_{k-2}=k+1, j_{k-1}= k+2, j_{k}=k+3\}$.
\end{enumerate}
On noting that $\mathbf{P}_{i}\mathbf{P}_{j}^{-1}$ is a matrix whose third power (i.e., cube) is the identity matrix (i.e., it is a permutation that can be decomposed into cycles of length $3$), its eigen values of the cube roots of $1$. This means that in a finite field whose size is a prime (which is not equal to $3$), its only unique eigen value is $1$. Note that this means that Property \ref{property:MDS} can be proved to hold in the first two scenarios using arguments similar to Case 1. For the third scenario, again, using arguments similar to Case 1, showing the MDS property is equivalent to showing that the matrix    $$ \left[\begin{array}{cc} 
\lambda_{i}\mathbf{P}_{i} &  \lambda_{j} \mathbf{P}_{j}\\
\lambda_{i}^{2}\mathbf{P}_{i}^{2} & \lambda_{j}^{2} \mathbf{P}_{j}^{2}
\end{array} \right]$$
has a full rank. The above matrix has a full rank because is equal to 
$$
\left[\begin{array}{cc} 
\mathbf{I} & \mathbf{I}\\
\lambda_{i}\mathbf{P}_{i} &  \lambda_{j} \mathbf{P}_{j}
\end{array} \right] \times \left[\begin{array}{cc} \lambda_{i} \mathbf{P}_{i}&0 \\ 0 &\lambda_{j} \mathbf{P}_{j} \end{array}\right]
$$
and both the matrices of the above product have full rank.
Now, for the fourth scenario, we need to show that the matrix
$$
\left[\begin{array}{cc} 
\mathbf{I} & \mathbf{I}\\
\lambda_{i}^{2}\mathbf{P}_{i}^{2} & \lambda_{j}^{2} \mathbf{P}_{j}^{2}
\end{array} \right]$$
has a full rank. This can be seen on noting that the determinant of the above matrix evaluates to 
$$|\lambda_j^{2} \mathbf{P}_{j}^{2} - \lambda_{i}^{2} \mathbf{P}_{i}^{2}|=\lambda_{j}^{2} |\mathbf{P}_{i}^{-2}||\mathbf{P}_{j}^{2} \mathbf{P}_{i}^{-2}-\lambda_{i}^{2}\lambda_{j}^{-2}\mathbf{I}|$$
which is non-zero if $\lambda_{i}^{2}\lambda_{j}^{-2} \neq 1,$ again because $\mathbf{P}_{j}^{2}\mathbf{P}_{i}^{-2}$ is a matrix whose eigen values are the cube roots of unity. The conditions in (\ref{eq:3parity}) ensures that $\lambda_{i}^{2} \neq \lambda_{j}^{2}$.
Finally, we consider to scenario 5 where we need to show that all the information can be recovered from $k-3$ systematic nodes, and all $3$ parity nodes. For this, we need
$$\left[\begin{array}{ccc}
\mathbf{I} & \mathbf{I}&\mathbf{I}\\
\lambda_{i}\mathbf{P}_{i} &  \lambda_{j} \mathbf{P}_{j}& \lambda_{l} \mathbf{P}_{l}\\
\lambda_{i}^{2}\mathbf{P}_{i}^{2} & \lambda_{j}^{2} \mathbf{P}_{j}^{2}& \lambda_{l}^{2} \mathbf{P}_{l}^{2}
\end{array} \right]$$
to have full rank. Note that the above matrix has a block Vandermonde structure, where each of the blocks commute pairwise because of Lemma \ref{lemma:commute}. This fact, combined with the fact that commuting block matrices can be expanded in a manner, similar to the element-wise determinant expansion, implies that the determinant of the above matrix is equal to  
$$\prod_{i,j}|\lambda_{i}\mathbf{P}_{i}-\lambda_{j}\mathbf{P}_{j}|$$
The determinant is non-zero since $\lambda_{i} \neq \lambda_{j}$ if $i \neq j$. This completes the proof of the desired MDS property.

\bibliographystyle{ieeetr}
\bibliography{Thesis}

\begin{thebibliography}{10}

\bibitem{Cadambe_Huang_Li_Permutation_ISIT}
V.~R. Cadambe, C.~Huang, and J.~Li, ``Permutation code: Optimal exact-repair of
  a single failed node in {MDS} code based distributed storage systems,'' {\em
  To appear in Proceedings of IEEE Symposium on Information Theory (ISIT)},
  July 2011.

\bibitem{Bruck_etal_ISIT}
I.~Tamo, Z.~Wang, and J.~Bruck, ``{MDS} array codes with optimal rebuilding,''
  {\em To appear in proceedings of IEEE Symposium on Information Theory
  (ISIT)}, July 2011.

\bibitem{Dimakis_Godfrey_Wainwright_Ramachandran}
A.~Dimakis, P.~Godfrey, M.~Wainwright, and K.~Ramchandran, ``Network coding for
  distributed storage systems,'' in {\em IEEE INFOCOM}, pp.~2000 --2008, may
  2007.

\bibitem{Wu_Dimakis}
Y.~Wu and A.~Dimakis, ``Reducing repair traffic for erasure coding-based
  storage via interference alignment,'' in {\em IEEE International Symposium on
  Information Theory}, pp.~2276 --2280, 28 2009-july 3 2009.

\bibitem{Shah_etal}
N.~B. Shah, K.~V. Rashmi, P.~V. Kumar, and K.~Ramachandran, ``Explicit codes
  minimizing repair bandwidth for distributed storage,'' {\em CoRR},
  vol.~abs/0908.2984, 2009.
\newblock http://arxiv.org/abs/0908.2984.

\bibitem{Suh_Ramachandran}
C.~Suh and K.~Ramchandran, ``Exact regeneration codes for distributed storage
  repair using interference alignment,'' {\em CoRR}, vol.~abs/1001.0107, 2010.
\newblock http://arxiv.org/abs/1001.0107.

\bibitem{Cadambe_Jafar_Maleki}
V.~R. Cadambe, S.~Jafar, and H.~Maleki, ``Distributed data storage with minimum
  storage regenerating codes - exact and functional repair are asymptotically
  equally efficient,'' {\em CoRR}, vol.~abs/1004.4299, April 2010.
\newblock http://arxiv.org/abs/1004.4299.

\bibitem{Suh_Ramachandran_general}
C.~Suh and K.~Ramchandran, ``On the existence of optimal exact-repair mds codes
  for distributed storage,'' {\em CoRR}, vol.~abs/1004.4663, April 2010.
\newblock http://arxiv.org/abs/1004.4663.

\bibitem{Gaston_Pujol_Circulant}
B.~Gaston and J.~Pujol, ``Double circulant minimum storage regenerating
  codes,'' {\em CoRR}, vol.~abs/1007.2401, 2010.
\newblock http://arxiv.org/abs/1007.2401.

\bibitem{Dimakis_Survey}
A.~G. Dimakis, K.~Ramchandran, Y.~Wu, and C.~Suh, ``A survey on network codes
  for distributed storage,'' {\em arxiv.org}, vol.~abs/1004.4438, 2010.
\newblock http://arxiv.org/abs/1004.4438.

\bibitem{Cullina_Dimakis_Ho}
D.~Cullina, A.~Dimakis, and T.~Ho, ``Searching for minimum storage regenerating
  codes,'' {\em Proceedings of the 47th Annual Allerton Conference on
  Communication, Control and Computation}, Sep 2009.
\newblock http://arxiv.org/abs/0910.2245.

\bibitem{Dimakis_EVENODD}
Z.~Wang, A.~G. Dimakis, and J.~Bruck, ``Rebuilding for array codes in
  distributed storage systems,'' {\em ACTEMT: Workshop on the Application of
  Communication Theory to Emerging Memory Technologies}, December 2010.
\newblock http://arxiv.org/abs/1009.3291.

\bibitem{RDP_repair}
L.~Xiang, Y.~Xu, J.~C. Lui, and Q.~Chang, ``Optimal recovery of single disk
  failure in rdp code storage systems,'' in {\em Proceedings of the ACM
  SIGMETRICS international conference on Measurement and modeling of computer
  systems}, SIGMETRICS '10, (New York, NY, USA), pp.~119--130, ACM, 2010.
\newblock http://doi.acm.org/10.1145/1811039.1811054.

\bibitem{Wu_Explicit}
Y.~Wu, ``Existence and construction of capacity-achieving network codes for
  distributed storage,'' in {\em IEEE International Symposium on Information
  Theory}, pp.~1150 --1154, 28 2009-july 3 2009.

\bibitem{Cadambe_Jafar_int}
V.~Cadambe and S.~Jafar, ``Interference alignment and the degrees of freedom of
  the k user interference channel,'' {\em IEEE Trans. on Information Theory},
  vol.~54, pp.~3425--3441, Aug. 2008.

\bibitem{Bresler_Tse_diversity}
G.~Bresler and D.~Tse, ``3 user interference channel: Degrees of freedom as a
  function of channel diversity,'' in {\em Communication, Control, and
  Computing, 2009. Allerton 2009. 47th Annual Allerton Conference on}, pp.~265
  --271, Oct. 2009.

\bibitem{Nazer_Gastpar_Jafar_Vishwanath}
B.~Nazer, M.~Gastpar, S.~A. Jafar, and S.~Vishwanath, ``Ergodic interference
  alignment,'' June 2009.

\bibitem{Suh_Tse_subspace}
C.~Suh and D.~Tse, ``Interference alignment for cellular networks,'' in {\em
  Communication, Control, and Computing, 2008 46th Annual Allerton Conference
  on}, pp.~1037 --1044, Sept. 2008.

\bibitem{Shah_insufficiency}
N.~B. Shah, K.~V. Rashmi, P.~V. Kumar, and K.~Ramchandran, ``Distributed
  storage codes with repair-by-transfer and non-achievability of interior
  points on the storage-bandwidth tradeoff,'' {\em CoRR}, vol.~abs/1011.2361,
  2010.
\newblock http://arxiv.org/abs/1011.2361.

\bibitem{Rouayheb_Ramchandran_Fractional}
S.~Y.~E. Rouayheb and K.~Ramchandran, ``Fractional repetition codes for repair
  in distributed storage systems,'' {\em CoRR}, vol.~abs/1010.2551, 2010.
\newblock http://arxiv.org/abs/1010.2551.

\bibitem{Yekhanin_Locally_Decodable}
S.~Yekhanin, ``Locally decodable codes,'' in {\em Now Publishers}, pp.~1878
  --1882, june 2010.

\bibitem{Bruck_etal}
I.~Tamo, Z.~Wang, and J.~Bruck, ``{MDS} array codes with optimal rebuilding,''
  {\em CoRR}, vol.~abs/1103.3737, 2011.
\newblock http://arxiv.org/abs/1103.3737.

\bibitem{EVENODD}
M.~Blaum, J.~Brady, J.~Bruck, and J.~Menon, ``Evenodd: an optimal scheme for
  tolerating double disk failures in raid architectures,'' in {\em Computer
  Architecture, 1994., Proceedings the 21st Annual International Symposium on},
  pp.~245 --254, Apr. 1994.

\bibitem{LIBERATION}
J.~S. Plank, ``The raid-6 liber8tion code,'' {\em The International Journal of
  High Performance Computing and Applications}, vol.~23, pp.~242--251, August
  2009.

\bibitem{RDP}
P.~Corbett, B.~English, A.~Goel, T.~Grcanac, S.~Kleiman, J.~Leong, and
  S.~Sankar, ``Row-diagonal parity for double disk failure correction,'' in
  {\em In Proceedings of the 3rd USENIX Symposium on File and Storage
  Technologies (FAST)}, pp.~1--14, 2004.

\bibitem{STAR}
C.~Huang and L.~Xu, ``Star : An efficient coding scheme for correcting triple
  storage node failures,'' {\em Computers, IEEE Transactions on}, vol.~57,
  pp.~889 --901, July 2008.

\bibitem{Dimitris_Dimakis_Hadamard}
D.~S. Papailiopoulos and A.~G. Dimakis, ``Distributed storage codes through
  hadamard designs,'' {\em To be presented in ISIT 2011}, july 2011.

\bibitem{Gomadam_Cadambe_Jafar}
K.~Gomadam, V.~Cadambe, and S.~Jafar, ``Approaching the capacity of wireless
  networks through distributed interference alignment,'' in {\em Submitted to
  Globecom 2008. Preprint available through the authors website}, March 2008.

\bibitem{Peters_Heath}
S.~Peters and R.~Heath, ``Interference alignment via alternating
  minimization,'' in {\em Acoustics, Speech and Signal Processing, 2009. ICASSP
  2009. IEEE International Conference on}, pp.~2445 --2448, April 2009.

\bibitem{Commuting_Matrices}
I.~Kovacs, D.~S. Silver, and S.~G. Williams, ``Determinants of commuting-block
  matrices,'' {\em The American Mathematical Monthly}, vol.~106, no.~10,
  pp.~pp. 950--952, 1999.
\newblock http://www.jstor.org/stable/2589750.

\end{thebibliography}

\end{document}